\documentclass{amsart}
\usepackage{amssymb}

\usepackage{amscd}
\usepackage{thmdefs}
\usepackage[dvips]{epsfig}

\input{tcilatex}
\theoremstyle{definition}
\theoremstyle{remark}
\numberwithin{equation}{section}

\input tcilatex

\begin{document}
\title[Graph Fractaloids]{Classification of Graph Fractaloids}
\author{Ilwoo Cho and Palle E. T. Jorgensen}
\address{Saint Ambrose Univ., Dep. of Math., 421 Ambrose Hall, 518 W. Locust St.,
Davenport, Iowa, 52803, U. S. A.\\
Univ. of Iowa, Dep. of Math., 14 McLean Hall, Iowa City, Iowa, 52242, U. S.
A.}
\email{chowoo@sau.edu\\
jorgen@math.uiowa.edu}
\thanks{The second named author is supported by the U. S. National Science
Foundation.}
\date{Dec., 2008}
\dedicatory{}
\thanks{}
\maketitle

\begin{abstract}
In this paper, we observe graph fractaloids, which are the graph groupoids
with fractal property. In particular, we classify them in terms of the
spectral data of certain Hilbert space operators, called the radial
operators. Based on these information, we can define the pair of two numbers 
$(N_{0},$ $N^{0}),$ for a given graph fractaloid $\Bbb{G},$ called the
fractal pair of $\Bbb{G}.$ The graph fractaloids are classified by such
pairs.
\end{abstract}

\strut \strut \strut \strut \strut

\section{Introduction}

The word ``fractaloids'' in the title may at first seem a bit puzzling; in
any case, calling for an explanation. We have chosen the terminology in
order to call attention to a certain feature in the study of analysis and
spectral theory on countable directed graphs. The idea is sketched briefly
below, and then taken up systematically again in Section 2, with precise
definitions.

While the area of spectral theory of countable directed graphs usually
refers to weighted graphs and then the spectrum of a suitable associated
graph Laplacian, our approach here builds instead on two different tools:
Starting with a given graph $G$ (vertices and edges) we build a groupoid
(morphisms from edges, etc. in the usual way), and we then construct an
associated von Neumann algebra $M_{G}.$ The von Neumann algebra construction
is reminiscent of von Neumann's original way of generating a ring operators
(alias von Neumann algebra) from a free group on a finite number of
generators. The feature the two constructions have in common is a set of
branching rules, and indeed these branching rules capture an essential
feature of fractals.

In fact, a given graph $G$ can be turned a ``symbol space'' for the kind of
fractals that are built from iterated function systems (IFSs). Here, we use
IFSs in the sense of Hutchinson (See [50]). i.e., spaces and measures\ built
from a repeated application of a given finite set of maps in an ambient
space, and a subsequent limit construction. (Two popular examples of IFSs in
very special cases are the familiar middle-third Cantor set, and the
Sierpinsky gasket.) Spectral analysis on $G$ then carries over to the IFS
fractal under consideration. Even for the familiar IFSs under current study,
spectral theory is not yet fully developed, and our use of the von Neumann
algebra $M_{G}$ adds some global invariants to the study of fractals.

The distinction between local and global is relevant when analysis or
spectral theory is considered for graphs, or more generally for infinite
systems, such as arise in statistical mechanics and thermodynamics. Erwin
Schrodinger, in his little book [51], illustrated this point with reference
to macroscopic laws vs. microscopic in the physics of diffusion. While bulk
diffusion as predicted by the heat equation is deterministic, it results by
contrast from taking limits of microscopic components (at the quantum
level), i.e., molecular movements. Hence, as Einstein noted (1905), in
explaining Brownian motions, the local theory may be modeled by (purely
statistical) random walk on discrete configurations, later to be widely
studied in the form of statistical graph models.

A second kind of ``fractal'' amenable to our von Neumann algebra approach
derives from a different family of iteration systems, again a von Neumann
construction, but now the objects are automata; i.e., the study of abstract
machines and problems they are able to solve. The study of automata is
related to formal language theory, understood as classes of formal languages
they can recognize. More specifically, an automaton is a mathematical model
for a finite state machine (FSM). Roughly, an input-output machine that,
given an input of symbols, then``jumps'' through a series of states
according to a transition function (expressed as a table).

Again, an automaton $\mathcal{A}$ arises as an iteration limit $L(\mathcal{A}%
),$ and we will study $L(\mathcal{A})$ with the use of our von Neumann
algebra $M_{G}.$

In both applications of $M_{G},$ we are taking advantage of a certain atomic
decomposition (developed in our paper) of $M_{G},$ and we explore its use in
the study of fractals in the two senses outlined above.

Our subject is at the cross roads of operator algebra and analysis of graphs
and fractals. As a result, in Section 2 below, we develop the basic tools we
will need from both subjects. This section includes careful definitions of
the concepts from both subjects. To make the paper more assessable, we take
the liberty of explaining and motivating the fundamental tools we need from
operator algebras so they make sense to researchers working on analysis of
countable directed graphs, and vice versa.

\strut The main purpose of this paper is to introduce a new algebraic
structures having certain fractal property, which is, sometimes, called 
\emph{fractality}. In particular, we are interested in the groupoidal
version of fractal groups. In [16], [19] and [20], we constructed (graph)
groupoids with fractal property, called \emph{fractaloid}s. And we
considered the spectral data of fractaloids in operator theoretical point of
view. In this paper, we observe the classification of graph fractaloids.

In [10] through [15], and [17] through [22], we introduced \emph{graph
groupoid}s induced by \emph{countable directed graphs}. A graph groupoid is
a categorial groupoid having its base, the set of all vertices. i.e., we can
regard all vertices as (multi-)units. Every groupoid having only one base
element is a group. So, if $G$ is a finite directed graph with its graph
groupoid $\Bbb{G},$ and if the vertex set $V(G)$ $\subset $ $\Bbb{G}$
consists of only one element, then the graph groupoid $\Bbb{G}$ is a group.
For example, if $G$ is the one-vertex-$n$-loop-edge graph, then the graph
groupoid $\Bbb{G}$ of $G$ is group-isomorphic to the free group $F_{n},$
with $n$-generators (See [10] and [11]). Notice that the free group $F_{n}$
is a fractal group (See [1]), for all $n$ $\in $ $\Bbb{N}.$ Remark that
every graph groupoid is a groupoid, but the converse does not hold in
general. So, our fractaloids may be partially understood in groupoid theory.
Therefore, to avoid the confusion, different from [19] and [20], we call
fractaloids (in the sense of [19] and [20]), \emph{graph fractaloids}, like
in [16]. In [19], we conjectured that the only ``connected,'' ``finite''
directed graphs, generating graph fractaloids, are graph-isomorphic to (i)
the one-vertex-multi-loop-edge graphs, or (ii) the one-flow circulant
graphs, or the certain connection of the previous kind of graphs. And, in
[16], this conjecture is solved. And the conclusion of the conjecture in
[16] shows that there are sufficiently many fractaloids, since we can find
sufficiently many ``finite'' directed graphs, generating graph fractaloids.
i.e., we have rich fractality on (graph) groupoids. We can have that the
connected finite directed graphs, generating graph fractaloids, are

(i)\ \ \ the one-vertex-multi-loop-edge graphs, or

(ii)\ \ the regularized graphs of the one-flow circulant graphs or the
shadowed graphs of them, or

(iii) the regularized graphs of the complete graphs or the shadowed graphs
of them, or

(iv)\ the regularized graphs of the vertex-fixed iterated glued graphs $G$ $%
\#^{v}$ $O_{n},$ where $G$ are one of the forms in (i) through (iv).

Again, the above conclusion shows that, even though we restrict our
interests to the case where we only consider graph fractaloids, generated by
a connected ``finite'' directed graphs, we have the rich fractaloidal
structures to handle. i.e., there are more connected finite directed graphs
what we expected in [19], which means good for the richness of fractaloids.

To detect the fractality of graph groupoids, we used automata theory in
[16], [19] and [20]: We found the ``automata-theoretical,'' and ``algebra''
characterization of graph fractaloids. In this paper, we avoid to use the
automata theory. However, our construction is completely based on automata
theory. Based on the characterizations of graph fractaloids in [19], we can
find the ``graph-theoretical'' characterization of graph fractaloids in
[16], and it leads us to define graph fractaloids without using automata
theory. Recall that, in [16], we show that: the graph groupoid $\Bbb{G}$ of
a connected locally finite directed graph $G$ is a graph fractaloid, if and
only if the out-degrees and the in-degrees of all vertices of $G$ are
identical from each other. So, by using this characterization, we can
re-define graph fractaloids as in Section 3, below.

As in [10] through [15], we construct a von Neumann algebra $\mathcal{M}%
_{G}, $ generated by the graph groupoid $\Bbb{G}$ of $G,$ as a groupoid $%
W^{*}$-algebra $vN(L(\Bbb{G}))$ generated by the graph groupoid $\Bbb{G}$ in 
$B(H_{G}),$ where $(H_{G},$ $L)$ is the canonical (left) representation of $%
\Bbb{G},$ consisting of a suitable Hilbert space $H_{G},$ and the groupoid
action $L$ of $\Bbb{G},$ acting on $H_{G}.$ We call $\mathcal{M}_{G},$ \emph{%
the }(\emph{left})\emph{\ graph von Neumann algebra of} $G.$ In [16], [19],
and [20], we use the \emph{right} graph von Neumann algebra $M_{G}$ $=$ $%
vN(R(\Bbb{G}))$ of $G$ in $B(H_{G}),$ where $(H_{G},$ $R)$ is the canonical
``right'' representation of $\Bbb{G},$ where $R$ is the right action of $%
\Bbb{G},$ acting on $H_{G}.$

The right graph von Neumann algebras $M_{G}$ are the opposite $W^{*}$%
-algebras $\mathcal{M}_{G}^{op}$ of the graph von Neumann algebras $\mathcal{%
M}_{G}$. Thus the right graph von Neumann algebras $M_{G}$ and the graph von
Neumann algebras $\mathcal{M}_{G}$ are anti-$*$-isomorphic from each other.
The only difference is the choice of actions of a graph groupoid. In this
paper, we will use right graph von Neumann algebras, as in [19] and [20].

Let $G$ be a given connected locally finite directed graph with its graph
groupoid $\Bbb{G},$ and let $M_{G}$ the right graph von Neumann algebras of $%
G$. Then the graph groupoid $\Bbb{G}$ induces a certain Hilbert space
operator $T_{G}$ in $M_{G},$ called the \emph{labeling operator of} $\Bbb{G}$
(See [19] and [20]). It is self-adjoint in $M_{G}.$ It is known that the
free distributional data, represented by the $D_{G}$-valued (amalgamated or
operator-valued) free moments $\{E(T_{G}^{n})\}_{n=1}^{\infty }$ of $T_{G},$
contain the spectral information of $T_{G}$, where $D_{G}$ is the \emph{%
diagonal subalgebra} of $M_{G}$.

In [20], we found the general computations of $D_{G}$-valued free moments of 
$T_{G}$, and in [19], the spectral information of $T_{G}$ of graph
fractaloids $\Bbb{G}$ is completely characterized by computing the $D_{G}$%
-valued free moments: The computations are based on the observation of the
cardinalities of lattice paths with axis property (See [40]).

In this paper, under our new setting, we re-define the same operator $T_{G},$
called the \emph{radial operators} of the graph groupoid $\Bbb{G},$ as an
element of the right graph von Neumann algebra $M_{G}$ (See Section 4). By
definition, we can realize that the labeing operators in the sense of [19]
and [20], and our radial operators are equivalent. i.e., if a graph $G$ is
fixed, then the labeling operator and the radial operator are identically
distributed over $D_{G}$ in $B(H_{G}).$

A \emph{graph} is a set of objects called \emph{vertices} (or points or
nodes) connected by links called \emph{edges} (or lines). In a \emph{%
directed graph}, the two directions are counted as being distinct directed
edges (or arcs). A graph is depicted in a diagrammatic form as a set of dots
(for vertices), jointed by curves (for edges). Similarly, a directed graph
is depicted in a diagrammatic form as a set of dots jointed by arrowed
curves, where the arrows point the direction of the directed edges.

In this paper, we consider direct graph $G$ as a combinatorial pair $(V(G),$ 
$E(G)),$ where $V(G)$ is the vertex set of $G$ and $E(G)$ is the edge set of 
$G.$ As we assumed at the beginning of the paper, throughout this paper,
every graph is a locally finite countably directed graph. Equivalently, the
degree of $v$ $\in $ $V(G)$ is finite, for all $v$ $\in $ $V(G).$ Notice
that, since $G$ is directed (or oriented on $E(G)$), each edge $e$ has its 
\emph{initial vertex} $v_{1}$ and its \emph{terminal vertex} $v_{2}.$ i.e., $%
e$ connects from $v_{1}$ to $v_{2}.$ Remark that the vertices $v_{1}$ and $%
v_{2}$ are not necessarily distinct, in general; for instance, if $e$ is a
loop edge, then $v_{1}$ $=$ $v_{2}.$

Recall that the \emph{degree} $\deg (v)$ of a vertex $v$ is defined to be
the sum of the \emph{out-degree} $\deg _{out}(v)$ and the\emph{\ in-degree} $%
\deg _{in}(v),$ dependent upon the direction on $G.$ i.e.,

\begin{center}
$\deg (v)$ $\overset{def}{=}$ $\deg _{out}(v)$ $+$ $\deg _{in}(v)$
\end{center}

where\strut

\begin{center}
$\deg _{out}(v)$ $\overset{def}{=}$ $\left| \{e\in E(G):e\text{ has its
initial vertex }v\}\right| $
\end{center}

and

\begin{center}
$\deg _{in}(v)$ $\overset{def}{=}$ $\left| \{e\in E(G):e\text{ has its
terminal vertex }v\}\right| .$
\end{center}

\strut

Define now the number $N$ by$\strut $

\begin{center}
$N$ $\overset{def}{=}$ $\max $ $\{\deg _{out}(v)$ $:$ $v$ $\in $ $V(\widehat{%
G})$ $=$ $V(G)\}.$\strut
\end{center}

Notice that, by the locally finiteness of $G,$ $N$ $<$ $\infty $ in $\Bbb{N}%
. $

The main purpose of this paper is to classify the graph fractaloids, in
temrs of their spectral information. In [19], we showed that the free
distribution of the labeling operators (and hence, that of the radial
operators) of graph fractaloids are scalar-valued:

\begin{center}
$E(T_{G}^{n})$ $=$ $\gamma _{n}\cdot 1_{D_{G}},$ for all $n$ $\in $ $\Bbb{N}$%
,
\end{center}

where $T_{G}$ is the labeling operator of a graph fractaloid $\Bbb{G}$ in
the right graph von Neumann algebra $M_{G}$, and where $\gamma _{n}$ is the
cardinality of a certain subset of the collection of all lattice paths
induced by $N$-lattices.

The above free-moment computations show that if two connected locally finite
directed graphs $G_{1}$ and $G_{2}$ have graph-isomorphic shadowed graphs,
then the corresponding graph groupoids $\Bbb{G}_{1}$ and $\Bbb{G}_{2}$ are
groupoid-isomorphic; and if $\Bbb{G}_{k}$ are graph fractaloids, for $k$ $=$ 
$1,$ $2,$ then the radial operators $T_{G_{1}}$ and $T_{G_{2}}$ are
identically distributed over $D_{G}$ in $M_{G}.$

Therefore we can determine the classification of graph fractaloids in terms
of their spectral property, with respect to the \emph{fractal pairs},
consisting of the certain numbers.

Let

\begin{center}
$n$ $=$ $\max \{\deg _{out}(v)$ $:$ $v$ $\in $ $V(G)\}$ $\in $ $\Bbb{N},$ in 
$G,$
\end{center}

and

\begin{center}
$m$ $=$ $\left| V(G)\right| $ $\in $ $\Bbb{N}_{\infty }$ $\overset{def}{=}$ $%
\Bbb{N}$ $\cup $ $\{\infty \}.$
\end{center}

Then the pair $(n,$ $m)$ is well-determined, whenever we have a connected
``locally finite'' directed graph $G.$ If $G$ generates a graph fractaloid $%
\Bbb{G},$ then this pair $(n,$ $m)$ is called the \emph{fractal pair of} $%
\Bbb{G}.$ Our main result of this paper is that the given two graph
fractaloids $\Bbb{G}_{1}$ and $\Bbb{G}_{2}$ have the same fractal pair $%
(N_{0},$ $N^{0}),$ then the radial operators $T_{G_{1}}$ and $T_{G_{2}}$ of
them are identically free distributed over $\Bbb{C}^{\oplus \,N^{0}}.$ In
particular,

\begin{center}
$E(T_{G_{k}}^{n})$ $=$ $\left| \mathcal{L}_{N_{0}}^{o}(n)\right| $ $\cdot $ $%
1_{\Bbb{C}^{\oplus \,N^{0}}},$ for all $n$ $\in $ $\Bbb{N}$,
\end{center}

where $\mathcal{L}_{N_{0}}^{o}(n)$ is the set consisting of all length-$%
N_{0} $ lattice paths in $\Bbb{R}^{2},$ starting at $(0,$ $0),$ and ending
on the hrozontal axis. These fractal pairs on the set $\mathcal{F}_{ractal}$
of all graph fractaloids make us classify the set $\mathcal{F}_{ractal},$ as
follows:

\begin{center}
$\mathcal{F}_{ractal}$ $=$ $\underset{(n,\,m)\in \Bbb{N}\times \Bbb{N}%
_{\infty }}{\sqcup }$ $\left( [(n,\text{ }m)]\right) ,$
\end{center}

where $[(n,$ $m)]$ is an equivalence class in $\mathcal{F}_{ractal}.$\strut

\section{Background and Definitions\strut \strut \strut \strut \strut \strut
\strut \strut \strut \strut \strut \strut \strut}

\strut Recently, countable directed graphs have been studied in Pure and
Applied Mathematics, because not only that they are involved by a certain
noncommutative structures but also that they visualize such structures.
Futhermore, the visualization has nice matricial expressions, (sometimes,
the operator-valued matricial expressions dependent on) adjacency matrices
or incidence matrices of the given graph. In particular, partial isometries
on a Hilbert space can be expressed and visualized by directed graphs: in
[10] through [15], [17], and [23], we have seen that each edge (resp. each
vertex) of a graph corresponds to a partial isometry (resp. a projection) on
a Hilbert space. In [18], [21], and [22], we showed that any finite partial
isometries (and the initial and final projections induced by these partial
isometries) on an arbitrary separable infinite Hilbert space induces a
(locally finite) directed graph. This shows that there are close relations
between Hilbert space operators and directed graphs.

Also, in this paper, we consider the property of fractaloids in terms of the
spectral property of certain operators on Hilbert spaces (Also, see
[19]).\strut To do that, in this section, we introduce the concepts we will
use in the rest of the context.\strut \strut

\subsection{Graph Groupoids and Representations\strut}

Let $G$ be a directed graph with its vertex set $V(G)$ and its edge set $%
E(G).$ Let $e$ $\in $ $E(G)$ be an edge connecting a vertex $v_{1}$ to a
vertex $v_{2}.$ Then we write $e$ $=$ $v_{1}$ $e$ $v_{2},$ for emphasizing
the initial vertex $v_{1}$ of $e$ and the terminal vertex $v_{2}$ of $e.$
For a graph $G,$ we can define the oppositely directed graph $G^{-1},$ with $%
V(G^{-1})$ $=$ $V(G)$ and $E(G^{-1})$ $=$ $\{e^{-1}$ $:$ $e$ $\in $ $E(G)\},$
where each element $e^{-1}$ satisfies that $e$ $=$ $v_{1}$ $e$ $v_{2}$ in $%
E(G)$, with $v_{1},$ $v_{2}$ $\in $ $V(G),$ if and only if $e^{-1}$ $=$ $%
v_{2}$ $e^{-1}$ $v_{1},$ in $E(G^{-1}).$ This opposite directed edge $e^{-1}$
$\in $ $E(G^{-1})$ of $e$ $\in $ $E(G)$ is called the \emph{shadow of} $e.$
Also, this new graph $G^{-1}$, induced by $G,$ is said to be the \emph{%
shadow of} $G.$ It is clear that $(G^{-1})^{-1}$ $=$ $G.$\strut

Define the \emph{shadowed graph} $\widehat{G}$ of $G$ by a directed graph
with its vertex set $V(\widehat{G})$ $=$ $V(G)$ $=$ $V(G^{-1})$ and its edge
set $E(\widehat{G})$ $=$ $E(G)$ $\cup $ $E(G^{-1})$, where $G^{-1}$ is the 
\emph{shadow} of $G$. We say that two edges $e_{1}$ $=$ $v_{1}$ $e_{1}$ $%
v_{1}^{\prime }$ and $e_{2}$ $=$ $v_{2}$ $e_{2}$ $v_{2}^{\prime }$ are \emph{%
admissible}, if $v_{1}^{\prime }$ $=$ $v_{2},$ equivalently, the finite path 
$e_{1}$ $e_{2}$ is well-defined on $\widehat{G}.$ Similarly, if $w_{1}$ and $%
w_{2}$ are finite paths on $G,$ then we say $w_{1}$ and $w_{2}$ are \emph{%
admissible}, if $w_{1}$ $w_{2}$ is a well-defined finite path on $G,$ too.
Similar to the edge case, if a finite path $w$ has its initial vertex $v$
and its terminal vertex $v^{\prime },$ then we write $w$ $=$ $v_{1}$ $w$ $%
v_{2}.$ Notice that every admissible finite path is a word in $E(\widehat{G}%
).$ Denote the set of all finite path by $FP(\widehat{G}).$ Then $FP(%
\widehat{G})$ is the subset of $E(\widehat{G})^{*},$ consisting of all words
in $E(\widehat{G}).$

We can construct the \emph{free semigroupoid} $\Bbb{F}^{+}(\widehat{G})$ of
the shadowed graph $\widehat{G},$ as the union of all vertices in $V(%
\widehat{G})$ $=$ $V(G)$ $=$ $V(G^{-1})$ and admissible words in $FP(%
\widehat{G}),$ with its binary operation, the \emph{admissibility}$.$
Naturally, we assume that $\Bbb{F}^{+}\Bbb{(}\widehat{G})$ contains the
empty word $\emptyset .$ Remark that some free semigroupoid $\Bbb{F}^{+}\Bbb{%
(}\widehat{G})$ of $\widehat{G}$ does not contain the empty word; for
instance, if a graph $G$ is a one-vertex-multi-edge graph, then the shadowed
graph $\widehat{G}$ of $G$ is also a one-vertex-multi-edge graph, and it
induces the free semigroupoid $\Bbb{F}^{+}(\widehat{G}),$ which does not
have the empty word. However, in general, if $\left| V(G)\right| $ $>$ $1,$
then $\Bbb{F}^{+}(\widehat{G})$ always contain the empty word. Thus, if
there is no confusion, we always assume the empty word $\emptyset $ is
contained in the free semigroupoid $\Bbb{F}^{+}(\widehat{G})$ of $\widehat{G}%
.$

By defining the \emph{reduction }(RR) on $\Bbb{F}^{+}(\widehat{G}),$ we can
construct the graph groupoid $\Bbb{G}.$ i.e., the \emph{graph groupoid }$%
\Bbb{G}$ is a set of all ``reduced'' words in $E(\widehat{G}),$ with the
inherited admissibility on $\Bbb{F}^{+}(\widehat{G}),$ where the \emph{%
reduction} (RR) on $\Bbb{G}$ is\strut

(RR)\qquad $\qquad \qquad w$ $w^{-1}$ $=$ $v$ and $w^{-1}w$ $=$ $v^{\prime
}, $\strut

for all $w$ $=$ $v$ $w$ $v^{\prime }$ $\in $ $\Bbb{G},$ with $v,$ $v^{\prime
}$ $\in $ $V(\widehat{G}).$ In fact, this graph groupoid $\Bbb{G}$ is indeed
a categorial groupoid with its base $V(\widehat{G})$ (See Section 2.2).

Construct the canonical representation of the given graph groupoid $\Bbb{G}.$
Let

\begin{center}
$H_{G}$ $\overset{def}{=}$ $\underset{w\in FP_{r}(\widehat{G})}{\oplus }$ $(%
\Bbb{C}$ $\cdot $ $\xi _{w})$\strut
\end{center}

be the Hilbert space with its Hilbert basis $\{\xi _{w}$ $:$ $w$ $\in $ $%
FP_{r}(\widehat{G})\},$ where\strut

\begin{center}
$FP_{r}(\widehat{G})$ $\overset{def}{=}$ $\Bbb{G}$ $\setminus $ $\left( V(%
\widehat{G})\text{ }\cup \text{ }\{\emptyset \}\right) .$\strut
\end{center}

We will call $H_{G},$ the \emph{graph Hilbert space} induced by the graph $%
G. $ Notice that the basis elements $\xi _{w}$'s satisfy the multiplication
rule;

\begin{center}
$\xi _{w_{1}}\xi _{w_{2}}$ $=$ $\xi _{w_{1}w_{2}},$ for all $w_{1},$ $w_{2}$ 
$\in $ $FP_{r}(\widehat{G}),$
\end{center}

with $\xi _{\emptyset }$ $\overset{def}{=}$ $0_{H_{G}}$ in $H_{G}.$ Also, we
have, for any $w$ $\in $ $FP_{r}(\widehat{G}),$

\begin{center}
$\xi _{w}\xi _{w^{-1}}$ $=$ $\xi _{ww^{-1}},$ and $\xi _{w^{-1}}\xi _{w}$ $=$
$\xi _{w^{-1}w}.$
\end{center}

By the reduction (RR), $ww^{-1}$ and $w^{-1}w$ are vertices in $V(\widehat{G}%
).$ This shows that naturally, we can determine the Hilbert space elements $%
\xi _{w},$ for all $w$ $\in $ $\Bbb{G}.$

Define now the groupoid action of $\Bbb{G},$ acting on $H_{G},$

\begin{center}
$R$ $:$ $\Bbb{G}$ $\rightarrow $ $B(H_{G})$
\end{center}

by

\begin{center}
$R(w)$ $\overset{def}{=}$ $R_{w},$ for all $w$ $\in $ $\Bbb{G},$
\end{center}

where

\begin{center}
$R_{w}$ $\xi _{w^{\prime }}$ $\overset{def}{=}$ $\xi _{w^{\prime }w},$ for
all $w,$ $w^{\prime }$ $\in $ $\Bbb{G}.$
\end{center}

i.e., the operator $R_{w}$ is the ``right'' multiplication operator with its
symbol $\xi _{w}$ on $H_{G},$ for all $w$ $\in $ $\Bbb{G}.$\strut

\begin{definition}
Let $H_{G}$ be the graph Hilbert space induced by a given graph $G,$ and let 
$R$ be the right action of the graph groupoid $\Bbb{G}$ of $G,$ acting on $%
H_{G},$ defined as above. Then the pair $(H_{G},$ $R)$ is called the
(canonical) right representation of $\Bbb{G}.$\strut
\end{definition}

\subsection{Categorial Groupoids and Groupoid Actions\strut \strut}

We say an algebraic structure $(\mathcal{X},$ $\mathcal{Y},$ $s,$ $r)$ is a 
\emph{(categorial) groupoid} if it satisfies that (i) $\mathcal{Y}$ $\subset 
$ $\mathcal{X},$ (ii) for all $x_{1},$ $x_{2}$ $\in $ $\mathcal{X},$ there
exists a partially-defined binary operation $(x_{1},$ $x_{2})$ $\mapsto $ $%
x_{1}$ $x_{2},$ for all $x_{1},$ $x_{2}$ $\in $ $\mathcal{X},$ depending on
the source map $s$ and the range map $r$ satisfying the followings;\strut

(ii-1) $x_{1}$ $x_{2}$ is well-determined, whenever $r(x_{1})$ $=$ $s(x_{2})$
and in this case, $s(x_{1}$ $x_{2})$ $=$ $s(x_{1})$ and $r(x_{1}$ $x_{2})$ $%
= $ $r(x_{2}),$ for $x_{1},$ $x_{2}$ $\in $ $\mathcal{X},$\strut

(ii-2) $(x_{1}$ $x_{2})$ $x_{3}$ $=$ $x_{1}$ $(x_{2}$ $x_{3})$, if they are
well-determined in the sense of (ii-1), for $x_{1},$ $x_{2},$ $x_{3}$ $\in $ 
$\mathcal{X},$\strut

(ii-3) if $x$ $\in $ $\mathcal{X},$ then there exist $y,$ $y^{\prime }$ $\in 
$ $\mathcal{Y}$ such that $s(x)$ $=$ $y$ and $r(x)$ $=$ $y^{\prime },$
satisfying $x$ $=$ $y$ $x$ $y^{\prime }$ (Here, the elements $y$ and $%
y^{\prime }$ are not necessarily distinct),\strut \strut

(ii-4) if $x$ $\in $ $\mathcal{X},$ then there exists a unique element $%
x^{-1}$ for $x$ satisfying $x$ $x^{-1}$ $=$ $s(x)$ and $x^{-1}$ $x$ $=$ $%
r(x).$\strut

Thus, every group is a groupoid $(\mathcal{X},$ $\mathcal{Y},$ $s,$ $r)$
with $\left| \mathcal{Y}\right| $ $=$ $1$ (and hence $s$ $=$ $r$ on $%
\mathcal{X}$). This subset $\mathcal{Y}$ of $\mathcal{X}$ is said to be the 
\emph{base of} $\mathcal{X}$. Remark that we can naturally assume that there
exists the \emph{empty element} $\emptyset $ in a groupoid $\mathcal{X}.$
The empty element $\emptyset $ means the products $x_{1}$ $x_{2}$ are not
well-defined, for some $x_{1},$ $x_{2}$ $\in $ $\mathcal{X}.$ Notice that if 
$\left| \mathcal{Y}\right| $ $=$ $1$ (equivalently, if $\mathcal{X}$ is a
group), then the empty word $\emptyset $ is not contained in the groupoid $%
\mathcal{X}.$ However, in general, whenever $\left| \mathcal{Y}\right| $ $%
\geq $ $2,$ a groupoid $\mathcal{X}$ always contain the empty word. So, if
there is no confusion, we will naturally assume that the empty element $%
\emptyset $ is contained in $\mathcal{X}.$\strut

It is easily checked that our graph groupoid $\Bbb{G}$ of a countable
directed graph $G$ is indeed a groupoid with its base $V(\widehat{G}).$
i.e., every graph groupoid $\Bbb{G}$ of a countable directed graph $G$ is a
groupoid $(\Bbb{G},$ $V(\widehat{G}),$ $s$, $r)$, where $s(w)$ $=$ $s(v$ $w)$
$=$ $v$ and $r(w)$ $=$ $r(w$ $v^{\prime })$ $=$ $v^{\prime },$ for all $w$ $%
= $ $v$ $w$ $v^{\prime }$ $\in $ $\Bbb{G}$ with $v,$ $v^{\prime }$ $\in $ $V(%
\widehat{G}).$ i.e., the vertex set $V(\widehat{G})$ $=$ $V(G)$ is a base of 
$\Bbb{G}.$\strut

Let $\mathcal{X}_{k}$ $=$ $(\mathcal{X}_{k},$ $\mathcal{Y}_{k},$ $s_{k},$ $%
r_{k})$ be groupoids, for $k$ $=$ $1,$ $2.$ We say that a map $f$ $:$ $%
\mathcal{X}_{1}$ $\rightarrow $ $\mathcal{X}_{2}$ is a \emph{groupoid
morphism} if (i) $f$ is a function, (ii) $f(\mathcal{Y}_{1})$ $\subseteq $ $%
\mathcal{Y}_{2},$ (iii) $s_{2}\left( f(x)\right) $ $=$ $f\left(
s_{1}(x)\right) $ in $\mathcal{X}_{2},$ for all $x$ $\in $ $\mathcal{X}_{1}$%
, and (iv) $r_{2}\left( f(x)\right) $ $=$ $f\left( r_{1}(x)\right) $ in $%
\mathcal{X}_{2},$ for all $x$ $\in $ $\mathcal{X}_{1}.$ If a groupoid
morphism $f$ is bijective, then we say that $f$ is a \emph{%
groupoid-isomorphism}, and the groupoids $\mathcal{X}_{1}$ and $\mathcal{X}%
_{2}$ are said to be \emph{groupoid-isomorphic}.\strut

Notice that, if two countable directed graphs $G_{1}$ and $G_{2}$ are \emph{%
graph-isomorphic}, via a graph-isomorphism $g$ $:$ $G_{1}$ $\rightarrow $ $%
G_{2},$ in the sense that (i) $g$ is bijective from $V(G_{1})$ onto $%
V(G_{2}),$ (ii) $g$ is bijective from $E(G_{1})$ onto $E(G_{2}),$ (iii) $%
g(e) $ $=$ $g(v_{1}$ $e$ $v_{2})$ $=$ $g(v_{1})$ $g(e)$ $g(v_{2})$ in $%
E(G_{2}),$ for all $e$ $=$ $v_{1}$ $e$ $v_{2}$ $\in $ $E(G_{1}),$ with $%
v_{1},$ $v_{2}$ $\in $ $V(G_{1}),$ then the graph groupoids $\Bbb{G}_{1}$
and $\Bbb{G}_{2}$ are groupoid-isomorphic. More generally, if two graphs $%
G_{1}$ and $G_{2}$ have graph-isomorphic shadowed graphs $\widehat{G_{1}}$
and $\widehat{G_{2}}, $ then $\Bbb{G}_{1}$ and $\Bbb{G}_{2}$ are
groupoid-isomorphic (See [10] and [11]).\strut \strut \strut

Let $\mathcal{X}$ $=$ $(\mathcal{X},$ $\mathcal{Y},$ $s,$ $r)$ be a
groupoid. We say that this groupoid $\mathcal{X}$ \emph{acts on a set }$Y$
if there exists a groupoid action $\pi $ of $\mathcal{X}$ such that $\pi (x)$
$:$ $Y$ $\rightarrow $ $Y$ is a well-determined function, for all $x$ $\in $ 
$\mathcal{X}.$ Sometimes, we call the set $Y,$ a $\mathcal{X}$\emph{-set}.
\strut

Let $\mathcal{X}_{1}$ $\subset $ $\mathcal{X}_{2}$ be a subset, where $%
\mathcal{X}_{2}$ $=$ $(\mathcal{X}_{2},$ $\mathcal{Y}_{2},$ $s,$ $r)$ is a
groupoid, and assume that $\mathcal{X}_{1}$ $=$ $(\mathcal{X}_{1},$ $%
\mathcal{Y}_{1},$ $s,$ $r),$ itself, is a groupoid, where $\mathcal{Y}_{1}$ $%
=$ $\mathcal{X}_{2}$ $\cap $ $\mathcal{Y}_{2}.$ Then we say that the
groupoid $\mathcal{X}_{1}$ is a \emph{subgroupoid} of $\mathcal{X}_{2}.$%
\strut \strut

Recall that we say that a countable directed graph $G_{1}$ is a \emph{%
full-subgraph} of a countable directed graph $G_{2},$ if\strut

\begin{center}
$E(G_{1})$ $\subseteq $ $E(G_{2})$
\end{center}

and

\begin{center}
$V(G_{1})$ $=$ $\{v$ $\in $ $V(G_{1})$ $:$ $e$ $=$ $v$ $e$ or $e$ $=$ $e$ $%
v, $ $\forall $ $e$ $\in $ $E(G_{1})\}.$\strut
\end{center}

Remark the difference between full-subgraphs and subgraphs: We say that $%
G_{1}^{\prime }$ is a \emph{subgraph} of $G_{2},$ if\strut

\begin{center}
$V(G_{1}^{\prime })$ $\subseteq $ $V(G_{2})$
\end{center}

and

\begin{center}
$E(G_{1}^{\prime })$ $=$ $\{e$ $\in $ $E(G_{2})$ $:$ $e$ $=$ $v_{1}$ $e$ $%
v_{2},$ for $v_{1},$ $v_{2}$ $\in $ $V(G_{1}^{\prime })\}.$\strut
\end{center}

We can see that the graph groupoid $\Bbb{G}_{1}$ of $G_{1}$ is a subgroupoid
of the graph groupoid $\Bbb{G}_{2}$ of $G_{2},$ whenever $G_{1}$ is a
full-subgraph of $G_{2}.$\strut

\subsection{Right Graph von Neumann Algebras}

In this section, we briefly introduce right graph von Neumann algebras of
the graphs. Frankly speaking, we will not consider such operator algebraic
structures in detail, here. However, to study the spectral property of our
fractaloids, we need the frameworks where the corresponding labeling
operators of fractaloids work. For more about groupoid topological algebras,
see [19], [20], [22], [24] and [26]. And, for more about free probability,
see [5], [10], [11], and [28].

\begin{definition}
Let $G$ be a graph with its graph groupoid $\Bbb{G},$ and let $(H_{G},$ $R)$
be the right representation of $\Bbb{G},$ in the sense of Section 2.1. Under
the representation $(H_{G},$ $R),$ define the groupoid $W^{*}$-algebra $%
M_{G} $ $=$ $\overline{\Bbb{C}[R(\Bbb{G})]}^{w}$ in $B(H_{G}),$ as a $W^{*}$%
-subalgebra. This groupoid $W^{*}$-algebra $M_{G}$ is called the right graph
von Neumann algebra of $G.$\strut \ Define a $W^{*}$-subalgebra $D_{G}$ of $%
M_{G}$ by

\begin{center}
$D_{G}\overset{def}{=}$ $\underset{v\in V(\widehat{G})}{\oplus }$ $\left( 
\Bbb{C}\cdot R_{v}\right) .$
\end{center}

It is called the diagonal subalgebra of $M_{G}.$
\end{definition}

\begin{remark}
In [10] through [14], we observed the (left) multiplication operators $L_{w}$%
's, for all $w$ $\in $ $\Bbb{G},$ instead of using right multiplication
operators $R_{w}$'s. Then we can define the (left) graph von Neumann algebra 
$M_{G}^{op}$ $=$ $\overline{\Bbb{C}[L(\Bbb{G})]}^{w}$ in $B(H_{G}),$ where $%
L $ $:$ $\Bbb{G}$ $\rightarrow $ $B(H_{G})$ is the left groupoid action of $%
\Bbb{G},$ acting on $H_{G}$, i.e., $L_{w}$ $\xi _{w^{\prime }}$ $\overset{def%
}{=}$ $\xi _{ww^{\prime }},$ for all $w,$ $w^{\prime }$ $\in $ $\Bbb{G}.$
Notice that $M_{G}^{op}$ and $M_{G}$ are anti-$*$-isomorphic. Thus they
share the fundamental properties (See [19]). Indeed, the von Neumann algebra 
$M_{G}^{op}$ is the opposite $*$-algebra of our right graph von Neumann
algebra $M_{G}$ of $G.$
\end{remark}

\strut Notice that, every element $x$ in the right graph von Neumann algebra 
$M_{G}$ of $G$ has its expression,

\begin{center}
\strut $x$ $=$ $\underset{w\in \Bbb{G}}{\sum }$ $t_{w}$ $R_{w},$ with $t_{w}$
$\in $ $\Bbb{C}.$
\end{center}

Let $D_{G}$ be the diagonal subalgebra of $M_{G}.$ Define the canonical
conditional expectation

\begin{center}
$E$ $:$ $M_{G}$ $\rightarrow $ $D_{G}$
\end{center}

by

\begin{center}
$E\left( \underset{w\in \Bbb{G}}{\sum }t_{w}R_{w}\right) $ $\overset{def}{=}$
$\underset{v\in V(\widehat{G})}{\sum }$ $t_{v}$ $R_{v},$
\end{center}

for all $\underset{w\in \Bbb{G}}{\sum }t_{w}$ $R_{w}$ $\in $ $M_{G}.$ Then
the pair $(M_{G},$ $E)$ is a $D_{G}$-valued $W^{*}$-probability space over $%
D_{G},$ in the sense of Voiculescu (See [5] and [28]).

\begin{definition}
The $D_{G}$-valued $W^{*}$-probability space $(M_{G},$ $E)$ is called the
graph $W^{*}$-probability space induced by the given graph $G.$
\end{definition}

By [10], [11], [19], and [20], we have the following two theorems.

\begin{theorem}
(See [10] and [11]) Let $M_{G}$ be the right graph von Neumann algebra of $%
G. $ Then it is $*$-isomorphic to the $D_{G}$-valued reduced free product
algebra $\underset{e\in E(G)}{*_{D_{G}}^{r}}$ $M_{e}$ of the $D_{G}$-free
blocks $M_{e},$ where $M_{e}$ $\overset{def}{=}$ $vN(\Bbb{G}_{e},$ $D_{G})$
in $B(H_{G}),$ where $\Bbb{G}_{e}$ are the subgroupoid of $\Bbb{G},$ induced
by $\{e,$ $e^{-1}\},$ for all $e$ $\in $ $E(G).$ $\square $
\end{theorem}

\begin{theorem}
\strut (See [11]) Let $M_{G}$ be the right graph von Neumann algebra of $G,$
and let $\underset{e\in E(G)}{*_{D_{G}}^{r}}$ $M_{e}$ be the $D_{G}$-valued
reduced free product algebra of $M_{e}$'s, which is $*$-isomorphic to $%
M_{G}, $ in $B(H_{G}).$

(1) If $e$ is a loop edge, then the corresponding $D_{G}$-free block $M_{e}$
is $*$-isomorphic to the group von Neumann algebra $L(\Bbb{Z})$, generated
by the infinite cyclic abelian group $\Bbb{Z},$ which is also $*$-isomorphic
to the $L^{\infty }$-algebra $L^{\infty }(\Bbb{T}),$ where $\Bbb{T}$ is the
unit circle in $\Bbb{C}.$

(2) If $e$ is a non-loop edge, then $M_{e}$ is $*$-isomorphic to the
matricial algebra $M_{2}(\Bbb{C}),$ consisting of all $(2$ $\times $ $2)$%
-matrices. $\square $
\end{theorem}

Also, we can have the following classification theorem, in terms of graph
theory.

\begin{theorem}
(See [11]) Let $G_{1}$ and $G_{2}$ be directed graphs and assume that the
shadowed graphs $\widehat{G_{1}}$ and $\widehat{G_{2}}$ are
graph-isomorphic. Then the graph von Neumann algebras $M_{G_{1}}$ and $%
M_{G_{2}}$ are $*$-isomorphic. $\square $
\end{theorem}

\strut \strut Unfortunately, the converse of the previous theorem is unknown
(e.g., [10], [11], and [49]).\strut

\section{Graph Trees and Graph Fractaloids}

In this section, we define the fractality on graph groupoids. Our ``new''
definition of graph fractaloids is based on the original automata
theoretical definition of graph fractaloids in the sense of [19]. In Section
3.1, we briefly introduce the automata theoretical definition of graph
fractaloids. And Section 3.2, we re-define graph fractaloids.

\subsection{Fractality on Graph Groupoids}

\emph{Automata theory} is the study of abstract machines, and we are using
it in the formulation given by von Neumann. It is related to the theory of
formal languages. In fact, automata may be thought of as the class of formal
languages they are able to recognize. In von Neumann's version, an automaton
is a finite state machine (FSM). i.e., a machine with input of symbols,
transitions through a series of states according to a transition function
(often expressed as a table). The transition function tells the automata
which state to go to next, given a current state and a current symbol. The
input is read sequentially, symbol by symbol, for example as a tape with a
word written on it, registered by the head of the automaton; the head moves
forward over the tape one symbol at a time. Once the input is depleted, the
automaton stops. Depending on the state in which the automaton stops, it is
said that the automaton either accepts or rejects the input. The set of all
the words accepted by the automaton is called the language of the automaton.
For the benefit for the readers, we offer the following references for the
relevant part of Automata Theory: [1], [33], [34], [35], [48] and [49].\strut

Let the quadruple $\mathcal{A}$ $=$ $(D,$ $Q,$ $\varphi ,$ $\psi )$ be
given, where $D$ and $Q$ are sets and\strut

\begin{center}
$\varphi $ $:$ $D$ $\times $ $Q$ $\rightarrow $ $Q$ \ \ \ and \ \ $\psi $ $:$
$D$ $\times $ $Q$ $\rightarrow $ $D$\strut
\end{center}

are maps. We say that $D$ and $Q$ are the (finite) alphabet and the state
set of $\mathcal{A},$ respectively and we say that $\varphi $ and $\psi $
are the output function and the state transition function, respectively. In
this case, the quadruple $\mathcal{A}$ is called an automaton. If the map $%
\psi (\bullet ,$ $q)$ is bijective on $D,$ for any fixed $q$ $\in $ $Q,$
then we say that the automaton $A$ is \emph{invertible}. Similarly, if the
map $\varphi (x,$ $\bullet )$ is bijective on $Q,$ for any fixed $x$ $\in $ $%
D,$ then we say that the automaton $\mathcal{A}$ is \emph{reversible}. If
the automaton $\mathcal{A}$ is both invertible and reversible, then $%
\mathcal{A}$ is said to be \emph{bi-reversible}.\strut

To help visualize the use of automata, a few concrete examples may help.
With some oversimplification, they may be drawn from the analysis and
synthesis of input / output models in Engineering, often referred to as
black box diagram: excitation variables, response variables, and
intermediate variables (e.g., see [52] and [53]).\strut

Roughly speaking, a ``\emph{undirected}''\emph{\ tree} is a connected
simplicial graph without loop finite paths. Recall that a (undirected) graph
is \emph{simplicial}, if the graph has neither loop-edges nor multi-edges
connecting distinct two vertices. A \emph{directed tree} is a connected
simplicial graph without loop finite paths. In particular, we say that a
directed tree $\mathcal{T}_{n}$ is a $n$-\emph{regular tree}, if $\mathcal{T}%
_{n}$ is rooted, one-flowed, infinite directed tree, having the same
out-degrees $n$ for all vertices (See Section 3.2, for details). For
example, the $2$-regular tree $\mathcal{T}_{2}$ can be depicted by\strut

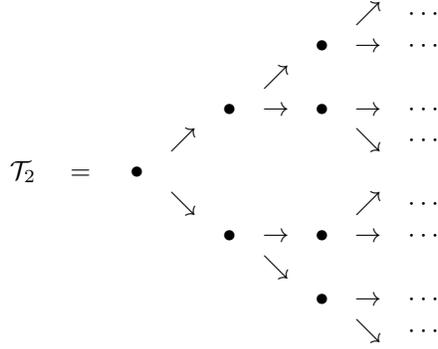
\begin{center}
$\mathcal{T}_{2}\quad =$\quad $
\begin{array}{lllllll}
&  &  &  &  & \nearrow & \cdots \\ 
&  &  &  & \bullet & \rightarrow & \cdots \\ 
&  &  & \nearrow &  &  &  \\ 
&  & \bullet & \rightarrow & \bullet & \rightarrow & \cdots \\ 
& \nearrow &  &  &  & \searrow & \cdots \\ 
\bullet &  &  &  &  &  &  \\ 
& \searrow &  &  &  & \nearrow & \cdots \\ 
&  & \bullet & \rightarrow & \bullet & \rightarrow & \cdots \\ 
&  &  & \searrow &  &  &  \\ 
&  &  &  & \bullet & \rightarrow & \cdots \\ 
&  &  &  &  & \searrow & \cdots
\end{array}
$\strut
\end{center}

Let $\mathcal{A}$ $=$ $(D,$ $Q,$ $\varphi ,$ $\psi )$ be an automaton with $%
\left| D\right| $ $=$ $n.$ Then, we can construct automata actions $\{%
\mathcal{A}_{q}$ $:$ $q$ $\in $ $Q\}$ of $\mathcal{A},$ acting on $\mathcal{T%
}_{n}.$ Let's fix $q$ $\in $ $Q.$ Then the action of $\mathcal{A}_{q}$ is
defined on the finite words $D_{*}$ of $D,$ by\strut

\begin{center}
$\mathcal{A}_{q}\left( x\right) $ $\overset{def}{=}$ $\varphi (x,$ $q),$ for
all $x$ $\in $ $D,$\strut
\end{center}

and recursively,$\strut $

\begin{center}
$\mathcal{A}_{q}\left( (x_{1},\text{ }x_{2},\text{ ..., }x_{m})\right) $ $=$ 
$\varphi \left( x_{1},\text{ }\mathcal{A}_{q}(x_{2},...,x_{m})\right) ,$%
\strut
\end{center}

for all $(x_{1},$ ..., $x_{m})$ $\in $ $D_{*},$ where\strut

\begin{center}
$D_{*}$ $\overset{def}{=}$ $\cup _{m=1}^{\infty }$ $\left( \left\{ (x_{1},%
\text{ ..., }x_{m})\in D^{m}\left| 
\begin{array}{c}
\text{ }x_{k}\in D,\text{ for all} \\ 
k=1,...,m
\end{array}
\right. \right\} \right) .$\strut
\end{center}

Then the automata actions $\mathcal{A}_{q}$'s are acting on the $n$-regular
tree $\mathcal{T}_{n}$. In other words, all images of automata actions are
regarded as an elements in the free semigroupoid $\Bbb{F}^{+}(\mathcal{T}%
_{n})$ of the $n$-regular tree. i.e.,\strut

\begin{center}
$V(\mathcal{T}_{n})$ $\supseteq $ $D_{*}$\strut
\end{center}

and its edge set\strut

\begin{center}
$\strut E(\mathcal{T}_{n})$ $\supseteq $ $\{\mathcal{A}_{q}(x)$ $:$ $x$ $\in 
$ $D,$ $q$ $\in $ $Q\}.$\strut
\end{center}

This makes us to illustrate how the automata actions work.\strut

Let $\mathcal{C}$ $=$ $\{\mathcal{A}_{q}$ $:$ $q$ $\in $ $Q\}$ be the
collection of automata actions of the given automaton $\mathcal{A}$ $=$ $<D,$
$Q,$ $\varphi ,$ $\psi >$. Then we can create a group $G(\mathcal{A})$
generated by the collection $\mathcal{C}.$ This group $G(\mathcal{A})$ is
called the automata group generated by $\mathcal{A}.$ The generator set $%
\mathcal{C}$ of $G(\mathcal{A})$ acts \emph{fully} on the $\left| D\right| $%
-regular tree $\mathcal{T}_{\left| D\right| },$ we say that this group $G(%
\mathcal{A})$ is a fractal group. There are many ways to define fractal
groups, but we define them in the sense of automata groups. (See [1] and
[35]. In fact, in [35], Batholdi, Grigorchuk and Nekrashevych did not define
the term ``fractal'', but they provide the fractal properties.)\strut

Now, we will define a fractal group more precisely (Also see [1]). Let $%
\mathcal{A}$ be an automaton and let $\Gamma $ $=$ $G(\mathcal{A})$ be the
automata group generated by the automata actions acting on the $n$-regular
tree $\mathcal{T}_{n},$ where $n$ is the cardinality of the alphabet of $%
\mathcal{A}.$ By $St_{\Gamma }(k),$ denote the subgroup of $\Gamma $ $=$ $G(%
\mathcal{A})$, consisting of those elements of $\Gamma ,$ acting trivially
on the $k$-th level of $\mathcal{T}_{n},$ for all $k$ $\in $ $\Bbb{N}$ $\cup 
$ $\{0\}.$\strut

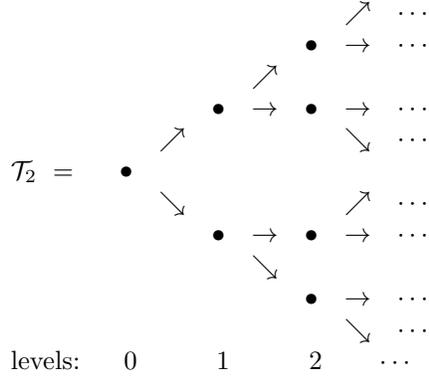
\begin{center}
$
\begin{array}{ll}
\mathcal{T}_{2}\text{ }= & 
\begin{array}{lllllll}
&  &  &  &  & \nearrow & \cdots \\ 
&  &  &  & \bullet & \rightarrow & \cdots \\ 
&  &  & \nearrow &  &  &  \\ 
&  & \bullet & \rightarrow & \bullet & \rightarrow & \cdots \\ 
& \nearrow &  &  &  & \searrow & \cdots \\ 
\bullet &  &  &  &  &  &  \\ 
& \searrow &  &  &  & \nearrow & \cdots \\ 
&  & \bullet & \rightarrow & \bullet & \rightarrow & \cdots \\ 
&  &  & \searrow &  &  &  \\ 
&  &  &  & \bullet & \rightarrow & \cdots \\ 
&  &  &  &  & \searrow & \cdots
\end{array}
\\ 
\text{levels:} & \,\,\,\,0\qquad \quad 1\qquad \quad 2\qquad \cdots
\end{array}
$\strut
\end{center}

Analogously, for a vertex $u$ in $\mathcal{T}_{n},$ define $St_{\Gamma }(u)$
by the subgroup of $\Gamma ,$ consisting of those elements of $\Gamma ,$
acting trivially on $u.$ Then$\strut $

\begin{center}
$St_{\Gamma }(k)$ $=$ $\underset{u\,:\,\text{vertices of the }k\text{-th
level of }\mathcal{T}_{n}}{\cap }$ $\left( St_{\Gamma }(u)\right) .$\strut
\end{center}

For any vertex $u$ of $\mathcal{T}_{n},$ we can define the algebraic
projection $p_{u}$ $:$ $St_{\Gamma }(u)$ $\rightarrow $ $\Gamma .$\strut

\begin{definition}
Let $\Gamma $ $=$ $G(\mathcal{A})$ be the automata group given as above. We
say that this group $\Gamma $ is a fractal group if, for any vertex $u$ of $%
\mathcal{T}_{n},$ the image of the projection $p_{u}\left( St_{\Gamma
}(u)\right) $ is group-isomorphic to $\Gamma ,$ after the identification of
the tree $\mathcal{T}_{n}$ with its subtree $\mathcal{T}_{u}$ with the root $%
u.$\strut
\end{definition}

For instance, if $u$ is a vertex of the $2$-regular tree $\mathcal{T}_{2}$,
then we can construct a subtree $\mathcal{T}_{u},$ as follows:\strut

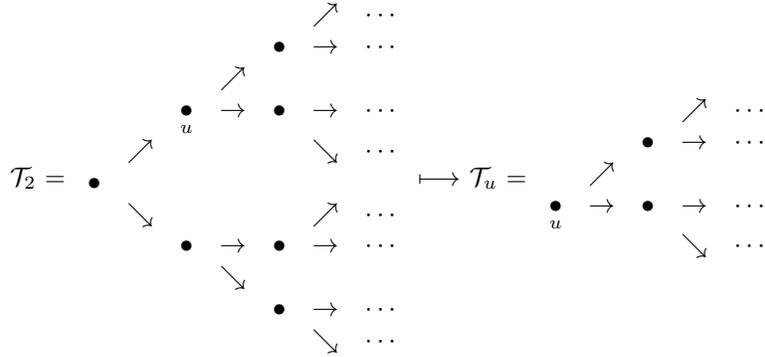
\begin{center}
$\mathcal{T}_{2}$ $=$ $
\begin{array}{lllllll}
&  &  &  &  & \nearrow & \cdots \\ 
&  &  &  & \bullet & \rightarrow & \cdots \\ 
&  &  & \nearrow &  &  &  \\ 
&  & \underset{u}{\bullet } & \rightarrow & \bullet & \rightarrow & \cdots
\\ 
& \nearrow &  &  &  & \searrow & \cdots \\ 
\bullet &  &  &  &  &  &  \\ 
& \searrow &  &  &  & \nearrow & \cdots \\ 
&  & \bullet & \rightarrow & \bullet & \rightarrow & \cdots \\ 
&  &  & \searrow &  &  &  \\ 
&  &  &  & \bullet & \rightarrow & \cdots \\ 
&  &  &  &  & \searrow & \cdots
\end{array}
$ $\longmapsto $ $\mathcal{T}_{u}$ $=$ $
\begin{array}{lllll}
&  &  & \nearrow & \cdots \\ 
&  & \bullet & \rightarrow & \cdots \\ 
& \nearrow &  &  &  \\ 
\underset{u}{\bullet } & \rightarrow & \bullet & \rightarrow & \cdots \\ 
&  &  & \searrow & \cdots
\end{array}
$\strut
\end{center}

As we can check, the graphs $\mathcal{T}_{2}$ and $\mathcal{T}_{u}$ are
graph-isomorphic. So, the above definition shows that if the automata
actions $\mathcal{A}_{q}$'s of $\mathcal{A}$ are acting \emph{fully} on $%
\mathcal{T}_{n},$ then the automata group $G(\mathcal{A})$ is a fractal
group.

The original definition of (graph) fractaloids in [19] are based on that of
fractal groups (Also, see [20] and [22]). To detect the fractality on a
connected locally finite directed graph $G$ (or the graph groupoid $\Bbb{G}$
of $G$), we define the corresponding automaton

\begin{center}
$\mathcal{A}_{G}$ $=$ $(\pm X,$ $E(\widehat{G}),$ $\varphi ,$ $\psi ),$
\end{center}

induced by $G.$ To do that we put the weight on $G$ (or the labeing on $G$)
by the labeling set $X.$ And then consider the automata actions $\{\mathcal{A%
}_{w}$ $:$ $w$ $\in $ $\Bbb{F}^{+}(\widehat{G})\}$: if the actions act fully
on the $2N$-regular tree $\mathcal{T}_{2N},$ then the groupoid $\Bbb{G}$ has
fractality, like fractal groups, where

\begin{center}
$N$ $=$ $\max \{\deg _{out}(v)$ $:$ $v$ $\in $ $V(G)\}$ $\in $ $\Bbb{N},$ in 
$G.$
\end{center}

Let $G$ be a given connected locally finite directed graph with the number $%
N,$ the maximum of the out-degrees of all vertices of $G.$ As usual, we
understand the real plane $\Bbb{R}^{2}$ as a 2-dimensional space generated
by the horizontal axis (or the $x$-axis) and the vertical axis (or the $y$%
-axis), which are homeomorphic to $\Bbb{R}.$ For the given number $N,$
define the lattices $l_{1},$ ..., $l_{N}$ in $\Bbb{R}^{2}$ by

\begin{center}
$l_{k}$ $\overset{def}{=}$ $\overrightarrow{(1,\text{ }e^{k})},$ for all $k$ 
$=$ $1,$ ..., $N,$
\end{center}

where $\overrightarrow{(t_{1},\text{ }t_{2})}$ means the vector connecting
the origin $(0,$ $0)$ to the point $(t_{1},$ $t_{2}),$ for $t_{1},$ $t_{2}$ $%
\in $ $\Bbb{R}.$ We call $l_{1},$ ..., $l_{N}$, the \emph{upward lattices for%
} $N.$ Define the set $X$ by the collection of all upward lattices for $N.$
i.e., $X$ $=$ $\{l_{1},$ ..., $l_{N}\}.$ With respect to the upward lattices 
$l_{1},$ ..., $l_{N},$ define the downward lattices $l_{-1},$ ..., $l_{-N}$
for $N,$ by

\begin{center}
$l_{-k}$ $\overset{def}{=}$ $\overrightarrow{(1,\text{ }-e^{k})},$ for all $%
k $ $=$ $1,$ ..., $N.$
\end{center}

Define the set $-X$ by the collection of all downward lattices for $N.$
i.e., $-X$ $=$ $\{l_{-1},$ ..., $l_{-N}\}.$ Define the set $\pm X$ by the
union of $X$ and $-X.$ i.e., $\pm X$ $=$ $X$ $\cup $ $-X.$ Then the set $\pm
X$ is called the \emph{labeling set of} $G$ (or $\Bbb{G}$).

For the given lattices in $\pm X,$ we can construct the \emph{lattice paths
in} $\Bbb{R}^{2}$ by the following ruls:

\begin{center}
$
\begin{array}{ll}
l_{i}l_{j}= & \text{the vector sum of }l_{i}\text{ and }l_{j} \\ 
& \text{by identifying the ending point }(1,\text{ }\varepsilon _{i}e^{i})
\\ 
& \text{of }l_{i}\text{ to the starting point }(0,\text{ }0)\text{ of }l_{j},
\end{array}
$
\end{center}

for all $i,$ $j$ $\in $ $\{\pm 1,$ ..., $\pm N\}$; inductively, we can
construct the lattice paths $l_{i_{1}}$ $l_{i_{2}}$ ... $l_{i_{n}},$ for all 
$n$ $\in $ $\Bbb{N},$ where $i_{1},$ ..., $i_{n}$ $\in $ $\{\pm 1,$ ..., $%
\pm N\}.$ Define the \emph{lattice path set} $\mathcal{L}_{N}$ \emph{%
generated by} $\pm X$ by the collectiong of all lattice paths defined as
above. Let $l$ $=$ $l_{i_{1}}$ ... $l_{i_{n}}$ $\in $ $\mathcal{L}_{N}.$
Then the length $\left| l\right| $ of $l$ is defined to be the number $n,$
the cardinality of the lattices generating the lattice path $l.$ So, the
lattice path set $\mathcal{L}_{N}$ is decomposed by

\begin{center}
$\mathcal{L}_{N}$ $=$ $\underset{k=1}{\overset{\infty }{\sqcup }}$ $\mathcal{%
L}_{N}(k),$
\end{center}

where

\begin{center}
$\mathcal{L}_{N}(k)$ $\overset{def}{=}$ $\{l$ $\in $ $\mathcal{L}_{N}$ $:$ $%
\left| l\right| $ $=$ $k\},$ for all $k$ $\in $ $\Bbb{N}.$
\end{center}

Clearly, $\pm X$ $=$ $\mathcal{L}_{N}(1),$ by definition.

Now, put the weights on edges of the shadowed graphs $\widehat{G}$ of $G$.
The \emph{weighting process on} $G$ is as follows:

(3.1.1) If $v$ $\in $ $V(G)$ and assume that $\deg _{out}(v)$ $=$ $k$ $\in $ 
$\Bbb{N}$, then

\begin{center}
$0$ $\leq $ $k$ $\leq $ $N,$ in $\Bbb{N}.$
\end{center}

Indeed, by the definition of $N,$ the out-degree $k$ $\leq $ $N$ in $\Bbb{N}%
. $ Now, let $e_{1},$ ..., $e_{k}$ be the edges in $E(G),$ having their
initial vertex $v,$ i.e., $e_{j}$ $=$ $v$ $e_{j},$ for all $j$ $=$ $1,$ ..., 
$k.$ Then, by the suitable re-arrange of these edges, we can give the \emph{%
lattice weights} $l_{j}$ $\in $ $X$ to the edges $e_{j}$, for all $j$ $=$ $%
1, $ ..., $k.$ Let's denote the weights of $e_{j}$'s by $\varpi (e_{j}),$
then, under our setting, $\varpi (e_{j})$ $=$ $l_{j},$ for all $j$ $=$ $1,$
..., $k.$ Do this process for all $v$ $\in $ $V(G).$ The graph $G$ with the
weighting process $\varpi $ is called the \emph{canonical labeled graph},
denoted by $(G,$ $\varpi ).$

(3.1.2) For the shadowed graph $G^{-1}$ of $G,$ we do the similar process
like (3.1.1). But, in this time, we use the set $-X,$ instead of $X.$ More
precisely, if $e$ $\in $ $E(G)$ with $\varpi (e)$ $=$ $l_{j}$ $\in $ $X,$
for $j$ $\in $ $\{1,$ ..., $k\},$ then take the weight $\varpi (e^{-1})$ of
the shadow $e^{-1}$ of $e$ by $l_{-j}$ $\in $ $-X.$ The pair $(G^{-1},$ $%
\varpi )$ is said to be the \emph{canonical labeled shadow of} $(G,$ $\varpi
).$

(3.1.3) The shadowed graph $\widehat{G}$ of $G$ can have the weighting
process based on (3.1.1) and (3.1.2). i.e., if $e$ $\in $ $E(\widehat{G}),$
then the weight $\varpi (e)$ of $e$ is determined by (3.1.1), whenever $e$ $%
\in $ $E(G),$ and it is determined by (3.1.2), whenever $e$ $\in $ $%
E(G^{-1}).$ Recall that $E(\widehat{G})$ $=$ $E(G)$ $\sqcup $ $E(G^{-1}).$
The pair $(\widehat{G},$ $\varpi )$ is called the\emph{\ canonical labeled
shadowed graph of} $(G,$ $\varpi ).$

In the rest of this section, all connected locally finite directed graphs
(resp., their shadowed graphs) are canonically labeled by the labeling set $%
X $ (resp. $\pm X$), as in (3.1.1), (resp., (3.1.2), and (3.1.3)).

Let's denote the empty lattice by $\emptyset _{X}$. i.e.,

\begin{center}
$\emptyset _{X}$ $=$ $\overrightarrow{(0,\text{ }0)}$ $\in $ $\Bbb{R}^{2}.$
\end{center}

Define the sets $\pm X_{0}$ and $\pm X_{0}^{*}$ by

\begin{center}
$\pm X_{0}$ $\overset{def}{=}$ $\{\emptyset \}$ $\cup $ $X$ $\cup $ $(-X),$
\end{center}

\strut and

\begin{center}
$\pm X_{0}^{*}$ $\overset{def}{=}$ $\{\emptyset _{X}\}$ $\cup $ $\mathcal{L}%
_{N}.$
\end{center}

\strut Define the subset $E(\widehat{G})_{0}$ of the free semigroupoid $\Bbb{%
F}^{+}(\widehat{G})$ of $\widehat{G}$ by

\begin{center}
$E(\widehat{G})_{0}$ $\overset{def}{=}$ $E(\widehat{G})$ $\cup $ $%
\{\emptyset \}.$
\end{center}

Now, for the given graph $G,$ define the corresponding automaton $\mathcal{A}%
_{G}$ by

\begin{center}
$\mathcal{A}_{G}$ $=$ $(\pm X_{0},$ $E(\widehat{G})_{0},$ $\varphi ,$ $\psi
),$
\end{center}

satisfying that

\begin{center}
$\varphi (l,$ $e)$ $\overset{def}{=}$ $\left\{ 
\begin{array}{ll}
l & \text{if }\exists e_{o}\in E(\widehat{G}),\text{ s.t., }\varpi (e_{o})=l
\\ 
\emptyset _{X} & \text{otherwise.}
\end{array}
\right. $
\end{center}

and

\begin{center}
$\psi (l,$ $e)$ $\overset{def}{=}$ $\left\{ 
\begin{array}{ll}
e_{o} & \text{if }\varphi (l,\text{ }e)=l \\ 
\emptyset & \text{otherwise,}
\end{array}
\right. $
\end{center}

for all $l$ $\in $ $\pm X_{0}$ and $e$ $\in $ $E(\widehat{G})_{0},$ with

\begin{center}
$\varphi (\emptyset _{X},$ $e)$ $=$ $\emptyset _{X},$ for all $e$ $\in $ $E(%
\widehat{G})_{0}$
\end{center}

and

\begin{center}
$\psi (l,$ $\emptyset )$ $=$ $\emptyset ,$ for all $l$ $\in $ $\pm X_{0}.$
\end{center}

Such an automaton $\mathcal{A}_{G}$ is called the \emph{graph-automaton
induced by }$G$ (or, in short, the $G$-\emph{automaton}). Then we can
construct the automata actions $\{\mathcal{A}_{w}$ $:$ $w$ $\in $ $\Bbb{F}%
^{+}(\widehat{G})\},$ acting on $\pm X_{0}^{*},$ and it is easy to check
that they act on the $2N$-regular tree $\mathcal{T}_{2N},$ because all
elements of $\pm X_{0}^{*}$ can be embedded in $\mathcal{T}_{2N},$ in the
natural manner. Assume that $\mathcal{T}^{G}$ is a full-subgraph of $%
\mathcal{T}_{2N},$ where the automata actions of $\mathcal{A}_{G}$ act
``fully'' on. Then this full-subgraph $\mathcal{T}^{G}$ is called the \emph{%
automata tree of }$\mathcal{A}_{G}$ (or $\mathcal{A}_{G}$-\emph{tree}).

For any nonempty $\varphi (l,$ $w)$, for $l$ $\in $ $\mathcal{L}_{N},$ and $%
w $ $\in $ $FP_{r}(\widehat{G}),$ we can define the tree $\mathcal{T}_{w},$
where the automata actions

\begin{center}
$\{\mathcal{A}_{w^{\prime }}$ $:$ $w^{\prime }$ $=$ $ww^{\prime \prime },$
for $w^{\prime \prime }$ $\in $ $\Bbb{F}^{+}(\widehat{G})\}$
\end{center}

are acting on. We call the trees $\mathcal{T}_{w}$ the $w$-\emph{parts of} $%
\mathcal{T}_{2N}$, for all $w$ $\in $ $FP_{r}(\widehat{G}).$ The reason why
we call $\mathcal{T}_{w}$'s the $w$-parts is that they are full-subgraph of
the automata tree $\mathcal{T}^{G}$ of $\mathcal{T}_{2N}.$ Similar to the
definition of fractality on groups, we can define the fractality on graph
groupoids as follows.

\begin{definition}
Let $G$ be a connected locally finite canonical labeled graph and let $%
\mathcal{A}_{G}$ be the $G$-automaton. Then the graph groupoid $\Bbb{G}$ of $%
G$ is said to be a graph fractaloid, if all $w$-parts $\mathcal{T}_{w}$ are
graph-isomorphic to the $\mathcal{A}_{G}$-tree $\mathcal{T}^{G},$ for all $w$
$\in $ $FP_{r}(\widehat{G}).$
\end{definition}

\strut The above definition is a natural extension of fractality on groups
to that on graph fractaloids. Actually the above definition can be extended
to define the fractality on groupois with fractality. So, in [19], instead
of using the term ``graph fractaloids,'' we simply use the term
``fractaloids.'' However, we prefer to use the term graph fractaloids,
because all graph groupoids are groupoids, but the converse does not hold.
In [19], we found the following two characterizations of graph fractaloids.

The following theorem is the automata-theoretical characterization of graph
fractaloids.

\begin{theorem}
(See [19]) Let $G$ be a canonical labeled graph with

\begin{center}
$N$ $=$ $\max \{\deg _{out}(v)$ $:$ $v$ $\in $ $V(G)\}$ $\in $ $\Bbb{N}.$
\end{center}

and let $\mathcal{A}_{G}$ be the $G$-automaton. Then the graph groupoid $%
\Bbb{G}$ of $G$ is a graph fractaloid, if and only if the automata actions $%
\{\mathcal{A}_{w}$ $:$ $w$ $\in $ $\Bbb{F}^{+}(\widehat{G})\}$ act fully on
the $2N$-regular tree $\mathcal{T}_{2N}.$ $\square $
\end{theorem}

\strut The following theorem is the algebraic characterization of graph
fractaloids.

\begin{theorem}
(See [19]) Let $G$ be given as in the previous theorem. Then the graph
groupoid $\Bbb{G}$ of $G$ is a graph fractaloid, if and only if the $%
\mathcal{A}_{G}$-tree $\mathcal{T}^{G}$ is graph-isomorphic to the $2N$%
-regular tree $\mathcal{T}_{2N}.$ $\square $
\end{theorem}

\strut The above theorems in fact show the difference between fractaloids
(groupoids with fractality) and graph fractaloids. Motivated by the previous
theorems, without using the automata theory, we can re-define graph
fractaloids in Section 3.2. In the rest of this section, we introduce
several examples for graph fractaloids. For more interesting examples, see
[22].

\begin{example}
(1) Let $O_{N}$ be the one-vertex-$N$-loop-edge graph, for $N$ $\in $ $\Bbb{N%
}.$ Then the graph groupoid $\Bbb{O}_{N}$ is a graph fractaloid. Recall
that, in fact, $\Bbb{O}_{N}$ is a group, which is group-isomorphic to the
free group $F_{N}$ with $N$-generators. And the free groups are fractal
groups (See [1]).

(2) Let $K_{N}$ be the one-flow circulant graph with $N$-vertices with

\begin{center}
$V(K_{N})$ $=$ $\{v_{1},$ ..., $v_{N}\},$
\end{center}

and

\begin{center}
$E(K_{N})$ $=$ $\{e_{j}$ $=$ $v_{j}$ $e_{j}$ $v_{j+1}$ $:$ $j$ $=$ $1,$ ..., 
$N,$ with $v_{N+1}$ $\overset{def}{=}$ $v_{1}\}.$
\end{center}

Then the graph groupoid $\Bbb{K}_{N}$ of $K_{N}$ is a graph fractaloid.

(3) Let $L_{\infty }$ be the infinite linear graph, graph-isomorphic to

\begin{center}
$\cdot \cdot \cdot \rightarrow \bullet \rightarrow \bullet \rightarrow
\bullet \rightarrow \cdot \cdot \cdot $
\end{center}

Then the graph groupoid $\Bbb{L}_{\infty }$ is a graph fractaloid.\strut

(4) Let $C_{N}$ be the complete graph with $N$-vertices. Recall that we say
that a graph $G$ is complete, if, for any pair $(v_{1},$ $v_{2})$ of a
distinct vertices, there always exists an edge $e$ $\in $ $E(G),$ such that $%
e$ $=$ $v_{1}$ $e$ $v_{2}.$ Then the graph groupoid $\Bbb{G}(C_{N})$ of $%
C_{N}$ is a graph fractaloid.
\end{example}

In [16], we obtain the following graph-theoretical characterization of graph
fractaloids, induced by connected locally finite (finite or infinite)
directed graphs.

\begin{theorem}
(See [16]) Let $G$ be a connected locally finite directed graph with its
graph groupoid $\Bbb{G}.$ Then $\Bbb{G}$ is a graph fractaloid, if and only
if the out-degrees and the in-degrees of all vertices are identical in $G.$
i.e., a graph $G$ generates a graph fractaloid, if and only if

\begin{center}
$\deg _{out}(v)$ $=$ $N$ $=$ $\deg _{in}(v),$ in $G,$
\end{center}

for all $v$ $\in $ $V(G),$ where

\begin{center}
$N$ $=$ $\max \{\deg _{out}(v)$ $:$ $v$ $\in $ $V(G)\}.$
\end{center}

$\square $
\end{theorem}

\strut By the previous theorem, without using automata theory, we can define
the graph fractaloids in the following section. However, we want to
emphasize that the above theorem is proven in [16], thanks to the
automata-theoretical and algebraic characterization of graph fractaloids
obtained in [19], based on the automata-theoretical setting on graph
groupoids.

\subsection{Graph Fractaloids}

In this section, we construct the graph tree $\mathcal{T}_{G}$ induced by a
given connected locally finite directed graph $G.$ Throughout this section,
all graphs are automatically assumed to be connected, and locally finite.
Recall that a directed graph, having neither multi-edges nor loop finite
paths, is called a \emph{directed tree}. If a directed tree $G$ has at least
one vertex $v$, satisfying that $\deg _{in}(v)$ $=$ $0,$ is said to be a 
\emph{directed} \emph{tree} \emph{with root}(\emph{s}). The vertices with $0$
in-degree are called the roots of $G.$ Suppose we have a directer tree $G$
with roots, and assume that we fix one root $v_{0}.$ Then $G$ is called a 
\emph{rooted tree} with its root $v_{0}.$ Now, let $G$ be a rooted tree with
its root $v_{0},$ and assume that the direction of $G$ is one-flowed from
the root $v_{0}$ (equivalently, $v_{0}$ is the only root of $G$). Then $G$
is a \emph{one-flow rooted tree}. An one-flow rooted tree is infinite, then
it is said to be a \emph{growing rooted tree}. Assume that a growing rooted
tree $G$ satisfies that, for any $v$ $\in $ $V(G),$ the out-degree $\deg
_{out}(v)$ are all identical. Then $G$ is a\emph{\ regular tree}. In
particular, if $\deg _{out}(v)$ $=$ $N,$ for all $v$ $\in $ $V(G),$ then
this regular tree $G$ is called the $N$-\emph{regular tree}. To emphasize
the regularity of this tree $G,$ we denote this $N$-regular tree $G$ by $%
\mathcal{T}_{N}.$ For instance, the $2$-regular tree $\mathcal{T}_{2}$ is as
follows:

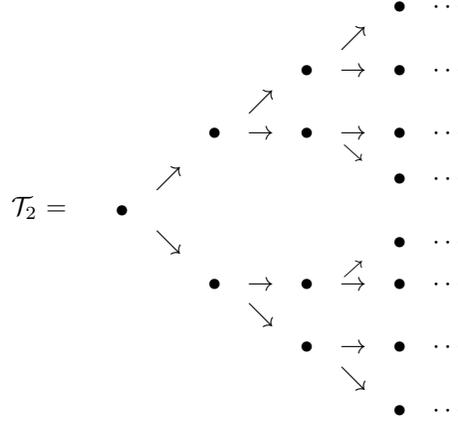
\begin{center}
$\mathcal{T}_{2}$ $=$ $
\begin{array}{lllllllll}
&  &  &  &  &  &  & \bullet & \cdots \\ 
&  &  &  &  &  & \nearrow &  &  \\ 
&  &  &  &  & \bullet & \rightarrow & \bullet & \cdots \\ 
&  &  &  & \nearrow &  &  &  &  \\ 
&  &  & \bullet & \rightarrow & \bullet & \underset{\searrow }{\rightarrow }
& \bullet & \cdots \\ 
&  & \nearrow &  &  &  &  & \bullet & \cdots \\ 
& \bullet &  &  &  &  &  &  &  \\ 
&  & \searrow &  &  &  &  & \bullet & \cdots \\ 
&  &  & \bullet & \rightarrow & \bullet & \overset{\nearrow }{\rightarrow }
& \bullet & \cdots \\ 
&  &  &  & \searrow &  &  &  &  \\ 
&  &  &  &  & \bullet & \rightarrow & \bullet & \cdots \\ 
&  &  &  &  &  & \searrow &  &  \\ 
&  &  &  &  &  &  & \bullet & \cdots
\end{array}
$
\end{center}

Let $G$ be a graph, and let

\begin{center}
$N$ $=$ $\max \{\deg _{out}(v)$ $:$ $v$ $\in $ $V(G)\}$ $<$ $\infty $ in $%
\Bbb{N}.$
\end{center}

Consider the shadowed graph $\widehat{G}$ of $G.$ Define the subsets $%
E_{v}^{v^{\prime }}$ of $E(\widehat{G})$ by

\begin{center}
$E_{v}^{v^{\prime }}$ $\overset{def}{=}$ $\{e$ $\in $ $E(\widehat{G})$ $:$ $%
e $ $=$ $v$ $e$ $v^{\prime }\},$
\end{center}

for all $(v,$ $v^{\prime })$ $\in $ $V(\widehat{G})^{2}.$ Remark that $v$
and $v^{\prime }$ are not necessarily distinct in $V(\widehat{G}).$ It is
possible that there exists a pair $(v_{1},$ $v_{2})$ of vertices such that $%
E_{v_{1}}^{v_{2}}$ is empty. By definition,

\begin{center}
$E(\widehat{G})$ $=$ $\underset{(v,v^{\prime })}{\cup }$ $E_{v}^{v^{\prime
}}.$
\end{center}

Then construct the graph tree $\mathcal{T}_{G}$ of $G,$ by re-arranging the
elements $V(\widehat{G})$ $\cup $ $E(\widehat{G}),$ up to the admissibility
on the free semigroupoid $\Bbb{F}^{+}(\widehat{G}),$ as follows. First fix
any arbitrary vertex $v_{0}$ $\in $ $V(\widehat{G})$ $=$ $V(G).$ Then
arrange $e$ $\in $ $\underset{v\in V(\widehat{G})}{\cup }$ $E_{v_{0}}^{v},$
by attaching them to $v_{0},$ preserving the direction on $G.$ i.e.,

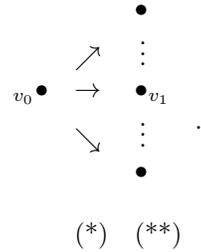
\begin{center}
$
\begin{array}{lll}
&  & \bullet \\ 
& \nearrow & \,\vdots \\ 
_{v_{0}}\bullet & \rightarrow & \bullet _{v_{1}} \\ 
& \searrow & \,\vdots \\ 
&  & \bullet \\ 
&  &  \\ 
& \text{(*)} & \text{(**)}
\end{array}
.$
\end{center}

Then we can have the above finite rooted tree with its root $v_{0}.$ Of
course, if the set $\underset{v\in V(\widehat{G})}{\cup }$ $E_{v_{0}}^{v}$
is empty, then we only have the trivial tree $G_{v_{0}},$ with $V(G_{v_{0}})$
$=$ $\{v_{0}\},$ and $E(G_{v_{0}})$ $=$ $\varnothing .$ The edges in the
column (*) is induced by the re-arrangement of the elements in $\underset{%
v\in V(\widehat{G})}{\cup }$ $E_{v_{0}}^{v},$ and the vertices in the column
(**) means the re-arrangement of the ``terminal'' vertices of the edges in $%
\underset{v\in V(\widehat{G})}{\cup }$ $E_{v_{0}}^{v}.$

Now, let $v_{1}$ $\in $ $V(\widehat{G})$ be an arbitrary chosen vertex of
the shadowed graph $\widehat{G}$ of $G,$ re-arranged in (**). Then we can do
the same process for $v_{1}.$ i.e., arrange the edges in $\underset{v\in V(%
\widehat{G})}{\cup }$ $E_{v_{1}}^{v}$ (if it is not empty), by attaching
them to $v_{1},$ preserving the direction on $G.$ i.e., we can construct

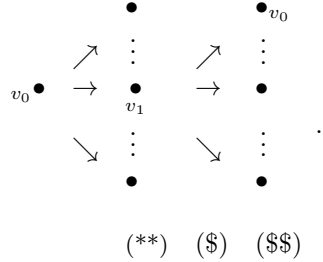
\begin{center}
$
\begin{array}{lllll}
&  & \bullet &  & \bullet _{v_{0}} \\ 
& \nearrow & \,\vdots & \nearrow & \,\vdots \\ 
_{v_{0}}\bullet & \rightarrow & \underset{v_{1}}{\bullet } & \rightarrow & 
\bullet \\ 
& \searrow & \,\vdots & \searrow & \,\vdots \\ 
&  & \bullet &  & \bullet \\ 
&  &  &  &  \\ 
&  & \text{{\small (**)}} & \text{(\$)} & \text{(\$\$)}
\end{array}
.$
\end{center}

\strut Here, the column (\$) is induced by the re-arrangement of the edges
in $\underset{v\in V(\widehat{G})}{\cup }$ $E_{v_{1}}^{v},$ and the vertices
in the column (\$\$) means the re-arrangement of the terminal vertices of
the edges in $\underset{v\in V(\widehat{G})}{\cup }$ $E_{v_{1}}^{v}.$ We can
do the same processes for all vertices in (**). Now, notice that it is
possible that one of the vertices in the columns (**) or (\$\$) can be $%
v_{0}.$ For instance, if $E_{v_{0}}^{v_{0}}$ is not empty (equivalently, if $%
v_{0}$ has an incident loop-edge), then $v_{0}$ is located in (**).
Similarly, $v_{0}$ can be located in (\$\$). For instance, if $v_{0}$ has
its incident length-2 loop finite path in $\Bbb{F}^{+}(\widehat{G}),$ then $%
v_{0}$ is in (\$\$). We admit such cases. i.e., a same vertex of $V(\widehat{%
G})$ can appear several times in this rooted-tree-making process.

Do this process inductively. If $G$ is infinite, then do this process
infinitely. The one-flow rooted tree, induced by this process, with its root 
$v_{0}$ is denoted by $\mathcal{T}_{v_{0}}.$ Notice that, from this process,
we can embed all elements (possibly finitely or infinitely many repeated
times) in $V(\widehat{G})$ $\cup $ $E(\widehat{G})$ into $\mathcal{T}%
_{v_{0}},$ preserving their admissibility! So, all elements in the free
semigroupoid $\Bbb{F}^{+}(\widehat{G})$ of the shadowed graph $\widehat{G}$
are embedded in $\mathcal{T}_{v_{0}}.$

\begin{definition}
Let $G$ be a connected locally finite directed graph with its shadowed graph 
$\widehat{G}.$ And let $\mathcal{T}_{v_{0}}$ be a rooted tree with its root $%
v_{0},$ induced by $G$. We say that this process is the graph-tree making of 
$G$. And the tree $\mathcal{T}_{v_{0}}$ is called the $v_{0}$-tree (or a
vertex-fixed graph tree) of $G.$
\end{definition}

\strut By definition, every connected locally finite directed graph $G$ has $%
\left| V(\widehat{G})\right| $-many vertex-trees of $G.$ Notice that the
vertex-trees of $G$ are determined by the vertices and edges in the
``shadowed'' graph $\widehat{G}$ of $G.$ The following proposition is easily
proven by the definition of the vertex-trees of a given graph, and by the
connectedness of our graphs.

\begin{proposition}
\strut Let $G$ be a connected locally finite directed graph with its
shadowed graph $\widehat{G}.$ Let $\Bbb{F}^{+}(\widehat{G})$ be the free
semigroupoid of $\widehat{G}.$ Then all elements in $\Bbb{F}^{+}(\widehat{G})
$ are embedded in the $v$-tree $\mathcal{T}_{v}$ of $G,$ for all $v$ $\in $ $%
V(\widehat{G})$ $=$ $V(G).$ $\square $
\end{proposition}

\strut \strut Observe now several examples for the construction of
vertex-graphs of a given graph.

\begin{example}
Let $O_{1}$ be a one-vertex-$1$-loop-edge graph with

\begin{center}
$V(O_{1})$ $=$ $\{v\}$ and $E(O_{1})$ $=$ $\{e$ $=$ $v$ $e$ $v\}.$
\end{center}

Then the shadowed graph $\widehat{O_{1}}$ of $O_{1}$ has its vertex set $V(%
\widehat{O_{1}})$, identical to $V(O_{1}),$ and its edge set

\begin{center}
$E(\widehat{O_{1}})$ $=$ $\{v\},$ and $E(\widehat{O_{1}})$ $=$ $\{e,$ $%
e^{-1}\}.$
\end{center}

Then we can construct the $v$-graph of $O_{1}$ by

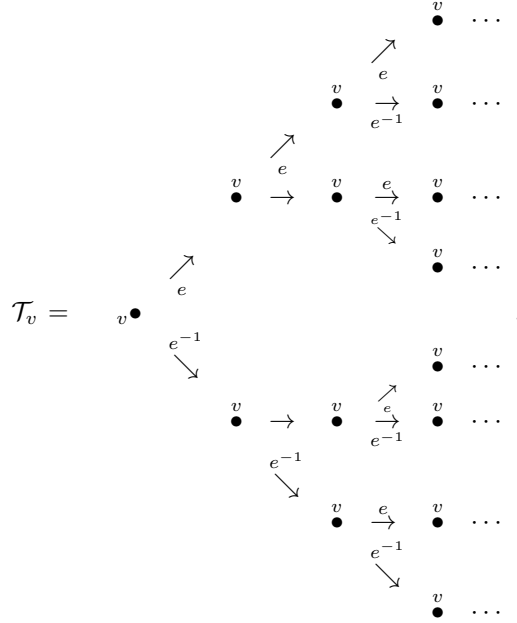
\begin{center}
$\mathcal{T}_{v}$ $=\quad 
\begin{array}{llllllll}
&  &  &  &  &  & \overset{v}{\bullet } & \cdots  \\ 
&  &  &  &  & \underset{e}{\nearrow } &  &  \\ 
&  &  &  & \overset{v}{\bullet } & \underset{e^{-1}}{\rightarrow } & 
\overset{v}{\bullet } & \cdots  \\ 
&  &  & \underset{e}{\nearrow } &  &  &  &  \\ 
&  & \overset{v}{\bullet } & \rightarrow  & \overset{v}{\bullet } & 
\underset{\overset{e^{-1}}{\searrow }}{\overset{e}{\rightarrow }} & \overset{%
v}{\bullet } & \cdots  \\ 
& \underset{e}{\nearrow } &  &  &  &  & \overset{v}{\bullet } & \cdots  \\ 
_{v}\bullet  &  &  &  &  &  &  &  \\ 
& \overset{e^{-1}}{\searrow } &  &  &  &  & \overset{v}{\bullet } & \cdots 
\\ 
&  & \overset{v}{\bullet } & \rightarrow  & \overset{v}{\bullet } & \overset{%
\underset{e}{\nearrow }}{\underset{e^{-1}}{\rightarrow }} & \overset{v}{%
\bullet } & \cdots  \\ 
&  &  & \overset{e^{-1}}{\searrow } &  &  &  &  \\ 
&  &  &  & \overset{v}{\bullet } & \overset{e}{\rightarrow } & \overset{v}{%
\bullet } & \cdots  \\ 
&  &  &  &  & \overset{e^{-1}}{\searrow } &  &  \\ 
&  &  &  &  &  & \overset{v}{\bullet } & \cdots 
\end{array}
.$
\end{center}

We can realize that the $v$-graph $\mathcal{T}_{v}$ is graph-isomorphic to
the $2$-regular graph $\mathcal{T}_{2}.$
\end{example}

\begin{example}
Let $G_{e}$ be the two-vertices-one-edge graph with

\begin{center}
$V(G_{e})$ $=$ $\{v_{1},$ $v_{2}\}$ and $E(G_{e})$ $=$ $\{e$ $=$ $v_{1}$ $e$ 
$v_{2}\}.$
\end{center}

Then the shadowed graph $\widehat{G_{e}}$ is a directed graph with

\begin{center}
$V(\widehat{G_{e}})$ $=$ $\{v_{1},$ $v_{2}\}$ and $E(\widehat{G_{e}})$ $=$ $%
\{e,$ $e^{-1}\}.$
\end{center}

So, we can have the $v_{1}$-tree $\mathcal{T}_{v_{1}}$ of $G,$

\begin{center}
$\mathcal{T}_{v_{1}}$ $=$ \quad $
\begin{array}{lllllllll}
_{v_{1}}\bullet  & \overset{e}{\rightarrow } & \underset{v_{2}}{\bullet } & 
\overset{e^{-1}}{\rightarrow } & \underset{v_{1}}{\bullet } & \overset{e}{%
\rightarrow } & \underset{v_{2}}{\bullet } & \overset{e^{-1}}{\rightarrow }
& \cdots 
\end{array}
,$
\end{center}

and the $v_{2}$-tree $\mathcal{T}_{v_{2}}$ of $G$,

\begin{center}
$\mathcal{T}_{v_{2}}$ $=$ \quad $
\begin{array}{lllllllll}
_{v_{2}}\bullet  & \overset{e^{-1}}{\rightarrow } & \underset{v_{1}}{\bullet 
} & \overset{e}{\rightarrow } & \underset{v_{2}}{\bullet } & \overset{e^{-1}%
}{\rightarrow } & \underset{v_{1}}{\bullet } & \overset{e}{\rightarrow } & 
\cdots 
\end{array}
.$
\end{center}

Therefore, both $\mathcal{T}_{v_{1}}$ and $\mathcal{T}_{v_{2}}$ are
graph-isomorphic to the $1$-regular tree $\mathcal{T}_{1}.$
\end{example}

\begin{example}
\strut Let $T_{2,1}$ be the finite tree with

\begin{center}
$V(T_{2,1})$ $=$ $\{v_{1},$ $v_{2},$ $v_{3}\}$
\end{center}

and

\begin{center}
$E(T_{2,1})$ $=$ $\{e_{1}$ $=$ $v_{1}$ $e_{1}$ $v_{2},$ $e_{2}$ $=$ $v_{1}$ $%
e_{2}$ $v_{3}\}.$
\end{center}

i.e.,

\begin{center}
$T_{2,1}$ $=$ $
\begin{array}{lll}
&  & \bullet _{v_{2}} \\ 
& \nearrow  &  \\ 
_{v_{1}}\bullet  &  &  \\ 
& \searrow  &  \\ 
&  & \bullet _{v_{3}}
\end{array}
.$
\end{center}

Then, after finding, the shadowed graph $\widehat{T_{2,1}}$ of $T_{2,1},$ we
can have the $v_{1}$-tree $\mathcal{T}_{v_{1}}$ of $T_{2,1},$

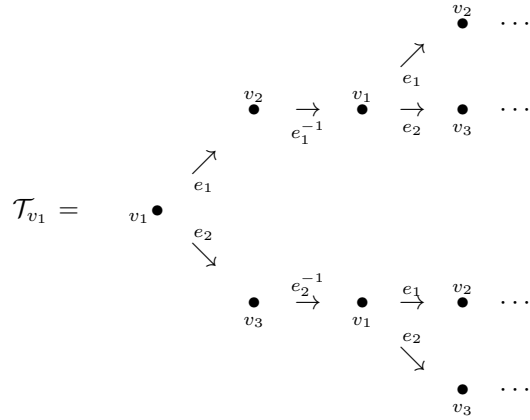
\begin{center}
$\mathcal{T}_{v_{1}}$ $=$ \quad $
\begin{array}{llllllll}
&  &  &  &  &  & \overset{v_{2}}{\bullet } & \cdots  \\ 
&  &  &  &  & \underset{e_{1}}{\nearrow } &  &  \\ 
&  & \overset{v_{2}}{\bullet } & \underset{e_{1}^{-1}}{\rightarrow } & 
\overset{v_{1}}{\bullet } & \underset{e_{2}}{\rightarrow } & \underset{v_{3}%
}{\bullet } & \cdots  \\ 
& \underset{e_{1}}{\nearrow } &  &  &  &  &  &  \\ 
_{v_{1}}\bullet  &  &  &  &  &  &  &  \\ 
& \overset{e_{2}}{\searrow } &  &  &  &  &  &  \\ 
&  & \underset{v_{3}}{\bullet } & \overset{e_{2}^{-1}}{\rightarrow } & 
\underset{v_{1}}{\bullet } & \overset{e_{1}}{\rightarrow } & \overset{v_{2}}{%
\bullet } & \cdots  \\ 
&  &  &  &  & \overset{e_{2}}{\searrow } &  &  \\ 
&  &  &  &  &  & \underset{v_{3}}{\bullet } & \cdots 
\end{array}
$
\end{center}

and the $v_{2}$-tree $\mathcal{T}_{v_{2}}$ of $T_{2,1},$

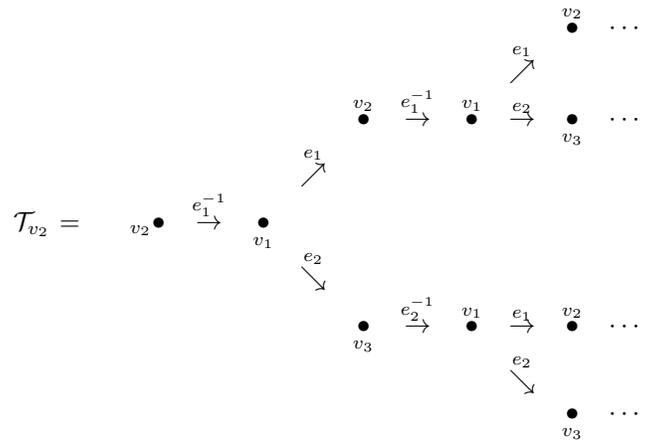
\begin{center}
$\mathcal{T}_{v_{2}}$ $=$ \quad $
\begin{array}{llllllllll}
&  &  &  &  &  &  &  & \overset{v_{2}}{\bullet } & \cdots  \\ 
&  &  &  &  &  &  & \overset{e_{1}}{\nearrow } &  &  \\ 
&  &  &  & \overset{v_{2}}{\bullet } & \overset{e_{1}^{-1}}{\rightarrow } & 
\overset{v_{1}}{\bullet } & \overset{e_{2}}{\rightarrow } & \underset{v_{3}}{%
\bullet } & \cdots  \\ 
&  &  & \overset{e_{1}}{\nearrow } &  &  &  &  &  &  \\ 
_{v_{2}}\bullet  & \overset{e_{1}^{-1}}{\rightarrow } & \underset{v_{1}}{%
\bullet } &  &  &  &  &  &  &  \\ 
&  &  & \overset{e_{2}}{\searrow } &  &  &  &  &  &  \\ 
&  &  &  & \underset{v_{3}}{\bullet } & \overset{e_{2}^{-1}}{\rightarrow } & 
\overset{v_{1}}{\bullet } & \overset{e_{1}}{\rightarrow } & \overset{v_{2}}{%
\bullet } & \cdots  \\ 
&  &  &  &  &  &  & \overset{e_{2}}{\searrow } &  &  \\ 
&  &  &  &  &  &  &  & \underset{v_{3}}{\bullet } & \cdots 
\end{array}
,$
\end{center}

and the $v_{3}$-graph of $T_{2,1},$

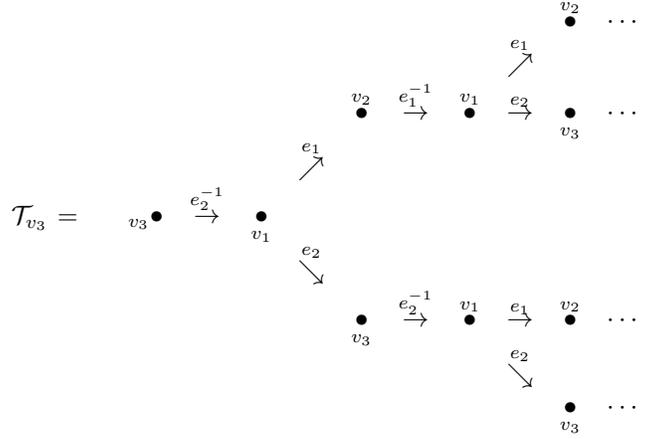
\begin{center}
$\mathcal{T}_{v_{3}}$ $=$ \quad $
\begin{array}{llllllllll}
&  &  &  &  &  &  &  & \overset{v_{2}}{\bullet } & \cdots  \\ 
&  &  &  &  &  &  & \overset{e_{1}}{\nearrow } &  &  \\ 
&  &  &  & \overset{v_{2}}{\bullet } & \overset{e_{1}^{-1}}{\rightarrow } & 
\overset{v_{1}}{\bullet } & \overset{e_{2}}{\rightarrow } & \underset{v_{3}}{%
\bullet } & \cdots  \\ 
&  &  & \overset{e_{1}}{\nearrow } &  &  &  &  &  &  \\ 
_{v_{3}}\bullet  & \overset{e_{2}^{-1}}{\rightarrow } & \underset{v_{1}}{%
\bullet } &  &  &  &  &  &  &  \\ 
&  &  & \overset{e_{2}}{\searrow } &  &  &  &  &  &  \\ 
&  &  &  & \underset{v_{3}}{\bullet } & \overset{e_{2}^{-1}}{\rightarrow } & 
\overset{v_{1}}{\bullet } & \overset{e_{1}}{\rightarrow } & \overset{v_{2}}{%
\bullet } & \cdots  \\ 
&  &  &  &  &  &  & \overset{e_{2}}{\searrow } &  &  \\ 
&  &  &  &  &  &  &  & \underset{v_{3}}{\bullet } & \cdots 
\end{array}
.$
\end{center}

We can check that $\mathcal{T}_{v_{2}}$ and $\mathcal{T}_{v_{3}}$ are
graph-isomorphic, but neither of them is graph-isomorphic to $\mathcal{T}%
_{v_{1}}.$
\end{example}

\begin{example}
\strut Let $K_{2}$ be the one-flow circulant graph with

\begin{center}
$V(K_{2})$ $=$ $\{v_{1},$ $v_{2}\},$
\end{center}

and

\begin{center}
$E(K_{2})$ $=$ $\{e_{1}$ $=$ $v_{1}$ $e_{1}$ $v_{2},\;$ $e_{2}$ $=$ $v_{2}$ $%
e_{2}$ $v_{1}\}.$
\end{center}

Then the shadowed graph $\widehat{K_{2}}$ of $K_{2}$ has

\begin{center}
$V(\widehat{K_{2}})$ $=$ $\{v_{1},$ $v_{2}\},$ and $E(\widehat{K_{2}})$ $=$ $%
\{e_{1}^{\pm 1},$ $e_{2}^{\pm 1}\}.$
\end{center}

By using the tree-making process, we obtain that the $v_{1}$-tree $\mathcal{T%
}_{v_{1}}$ and the $v_{2}$-tree $\mathcal{T}_{v_{2}}$ are graph-isomorphic
to the $2$-regular tree $\mathcal{T}_{2}.$ In general, every one-flow
circulant graph $K_{n}$ has its vertex-trees graph-isomorphic to the $2$%
-regular tree $\mathcal{T}_{2}.$
\end{example}

\strut As we have seen in the previous examples, sometimes, the vertex-trees
of a given graph are graph-isomorphic from each other, or not. In general,
the vertex-trees of a graph $G$ are not graph-isomorphic from each other.

\begin{definition}
Let $G$ be a connected locally finite directed graph and $\{\mathcal{T}_{v}$ 
$:$ $v$ $\in $ $V(\widehat{G})\}$, the collection of all vertex-trees of $G.$
Also, let

\begin{center}
$N$ $=$ $\max \{\deg _{out}(v)$ $:$ $v$ $\in $ $V(G)\}$ in $G$
\end{center}

(``not'' in $\widehat{G}$). If every $v$-tree $\mathcal{T}_{v}$ of $G$ is
graph-isomorphic to the $2N$-regular tree $\mathcal{T}_{2N},$ for all $v$ $%
\in $ $V(\widehat{G}),$ then the graph groupoid $\Bbb{G}$ of $G$ is called
the graph fractaloid induced by $G.$ And the graph $G$ is said to be a
fractal graph.
\end{definition}

\strut Under the above (new) definition of the fractal graphs and graph
fractaloids, we can re-obtain the graph-theoretical characterization of
graph fractaloids of [16]. In fact, the above new definition for graph
fractaloids (and fractal graphs) is based on the re-expression of automata
trees in the sense of [19]. The vertex-trees $\mathcal{T}_{v}$'s of a given
graph $G$ can be understood as the re-expression of the automata-trees
without using the automata-theoretical labeling process on $G$ (or on $\Bbb{G%
}$). Depeding on the new definition for graph fractaloids, we can get the
graph-theoretical characterization of graph fractaloids as follows:

\begin{theorem}
(See [16]) Let $G$ be a connected locally finite directed graph. The graph $G
$ is a fractal graph, if and only if

\begin{center}
$\deg _{out}(v)$ $=$ $\deg _{in}(v),$ in $G,$
\end{center}

for all $v$ $\in $ $V(G).$ $\square $
\end{theorem}

\strut \strut Thus, without loss of generality, we can re-define the fractal
graphs and graph fractaloids as follows: A connected locally finite directed
graph $G$ is a fractal graph, if the out-degrees and the in-degrees of all
vertices of $G$, in $G,$ are identical from each other. And, if a graph $G$
is a fractal graph, then the graph groupoid $\Bbb{G}$ of $G$ is said to be a
graph fractaloid.\strut 

\begin{example}
(1) The one-vertex-$n$-loop-edge graph $O_{n}$ is a fractal graph, for all $n
$ $\in $ $\Bbb{N}$, since 

\begin{center}
$\deg _{out}(v)$ $=$ $n$ $=$ $\deg _{in}(v),$ in $O_{n}$
\end{center}

where $v$ is the only vertex of $O_{n},$ for all $n$ $\in $ $\Bbb{N}.$

(2) The one-flow circulant graph $K_{n}$ is a fractal graph, for all $n$ $%
\in $ $\Bbb{N}$ $\setminus $ $\{1\},$ since

\begin{center}
$\deg _{out}(v)$ $=$ $1$ $=$ $\deg _{in}(v),$ in $K_{n},$
\end{center}

for all $v$ $\in $ $V(K_{n}).$

(3) Let $C_{n}$ be the complete graph with $n$-vertices, for $n$ $\in $ $%
\Bbb{N}$ $\setminus $ $\{1\}.$ i.e., it is a graph with

\begin{center}
$V(C_{n})$ $=$ $\{v_{1},$ ..., $v_{n}\},$
\end{center}

and

\begin{center}
$E(C_{n})$ $=$ $\left\{ e_{ij}\left| i\neq j\in \{1,...,n\}\right. \right\} ,
$
\end{center}

where $e_{ij}$ means the edge connecting the vertex $v_{i}$ to the vertex $%
v_{j}.$ Then the graph $C_{n}$ is a fractal graph, for all $n$ $\in $ $\Bbb{N%
}$ $\setminus $ $\{1\},$ since

$\deg _{out}(v_{j})$ $=$ $n$ $-$ $1$ $=$ $\deg _{in}(v_{j}),$ in $C_{n},$

for all $j$ $=$ $1,$ ..., $n.$

(4) Let $L$ be the infinite linear graph, graph-isomorphic to

\begin{center}
$\cdot \cdot \cdot \longrightarrow \bullet \longrightarrow \bullet
\longrightarrow \bullet \longrightarrow \cdot \cdot \cdot .$
\end{center}

Then it is a fractal graph, since

\begin{center}
$\deg _{out}(v)$ $=$ $1$ $=$ $\deg _{in}(v),$ in $L,$
\end{center}

for all $v$ $\in $ $V(L).$

(5) Let $\mathcal{T}_{k}$ be the $k$-regular graph, for $k$ $\in $ $\Bbb{N}.$
Then it is not a fractal graph. Assume that $v_{0}$ is a root of $\mathcal{T}%
_{k}.$ Then

\begin{center}
$\deg _{out}(v_{0})$ $=$ $k$ $\neq $ $0$ $=$ $\deg _{in}(v),$ in $\mathcal{T}%
_{k}.$
\end{center}

Therefore, the regular trees are not fractal.
\end{example}

More generally, we define graph fractaloids, without connectedness condition.

\begin{definition}
Let $G$ be a locally finite directed graph with its connected components $%
G_{1},$ ..., $G_{t},$ for $t$ $\in $ $\Bbb{N}$. Let $\Bbb{G}$ be the graph
groupoid of $G.$ We say that $\Bbb{G}$ is a graph fractaloid, if each $G_{j}$
generates a graph fractaloid $\Bbb{G}_{j},$ in the above sense, for all $j$ $%
=$ $1,$ ..., $t.$ In this case, the graph $G$ is called the ``disconnected''
fractal graph.
\end{definition}

However, in the rest of this paper, all our graphs are connected.

Let $G$ be a connected locally finite graph with its graph groupoid $\Bbb{G}.
$ Let $(v_{1},$ $v_{2})$ be the pair of vertices of $G$ (Remark that $v_{1}$
and $v_{2}$ are not necessarily distinct), and assume that there exists an
edge $e$ $=$ $v_{1}$ $e$ $v_{2}.$ Let's replace this edge $e$ to the $k$%
-multi-edges $e_{1},$ ..., $e_{k},$ satisfying $e_{j}$ $=$ $v_{1}$ $e_{j}$ $%
v_{2}.$ Do this process for all pair $(v,$ $v^{\prime })$ of the vertices of 
$G$, whenever there exists at least one edge connecting $v$ to $v^{\prime }.$
Clearly, if there is no edge connecting $v$ to $v^{\prime },$ then we do not
need to do this process. Then we can create a new connected locally finite
graph $G^{\prime },$ satisfying that

\begin{center}
$V(G^{\prime })$ $=$ $V(G).$
\end{center}

\begin{definition}
The new connected locally finite graph $G^{\prime }$ induced by a given
connected locally finite graph $G,$ in the previous paragraph, is called the
regularized graph of $G,$ denoted by $R_{k}(G),$ where $k$ is the
cardinality of the multi-edges in $G^{\prime }$ replaced by the edges in $G,$
for all $k$ $\in $ $\Bbb{N}.$ 
\end{definition}

\strut In [16], we showed that:

\begin{theorem}
(Also, see [16]) Let $G$ be a fractal graph. Then the $k$-regularized graph $%
R_{k}(G)$ is a fractal graph, too.
\end{theorem}

\begin{proof}
\strut Indeed, assume that $G$ is a fractal graph, satisfying that

\begin{center}
$\deg _{out}(v)$ $=$ $N$ $=$ $\deg _{in}(v),$ in $G,$
\end{center}

for all $v$ $\in $ $V(G).$ Then the $k$-regularized graph $R_{k}(G)$
satisfies that

\begin{center}
$\deg _{out}(x)$ $=$ $kN$ $=$ $\deg _{in}(x),$ in $R_{k}(G),$
\end{center}

for all $x$ $\in $ $V(R_{k}(G))$ $=$ $V(G),$ for all $k$ $\in $ $\Bbb{N}.$
Therefore, by the graph-theoretical characterization, the graph $R_{k}(G)$
is again a fractal graph. 
\end{proof}

\strut In the previous example, we showed that the graphs $O_{N},$ $K_{n},$ $%
C_{n}$ and $L$ are fractal graphs, for $N$ $\in $ $\Bbb{N},$ $n$ $\in $ $%
\Bbb{N}$ $\setminus $ $\{1\}.$ By the previous theorem, we can conclude that
the $k$-regularized graphs $R_{k}(O_{N}),$ $R_{k}(K_{n}),$ $R_{k}(C_{n}),$
and $R_{k}(L)$ are fractal graphs, too, for all $k$ $\in $ $\Bbb{N}.$ 

Let $G_{1}$ and $G_{2}$ be connected locally finite graphs. Define the \emph{%
unioned graph} $G$ $=$ $G_{1}$ $\cup $ $G_{2}$ of $G_{1}$ and $G_{2}$ by a
new directed graph with

\begin{center}
$V(G)$ $=$ $V(G_{1})$ $\cup $ $V(G_{2}),$
\end{center}

and

\begin{center}
$E(G)$ $=$ $E(G_{1})$ $\cup $ $E(G_{2}).$
\end{center}

So, every ``disjoint'' unioned graph $G,$ satisfying

\begin{center}
$V(G)$ $=$ $V(G_{1})$ $\sqcup $ $V(G_{2}),$
\end{center}

and

\begin{center}
$E(G)$ $=$ $E(G_{1})$ $\sqcup $ $E(G_{2})$,
\end{center}

is a unioned graph. But, notice that not all unioned graphs are disjoint
unioned graphs! For instance, if $G_{1}$ and $G_{2}$ are full-subgraphs of a
connected locally finite graph $K,$ then it is possible that

\begin{center}
$V(G_{1})$ $\cap $ $V(G_{2})$ $\neq $ $\varnothing ,$
\end{center}

or

\begin{center}
$E(G_{1})$ $\cap $ $E(G_{2})$ $\neq $ $\varnothing .$
\end{center}

Also, our shadowed graphs are unioned graphs which are not disjoint unioned
graphs. i.e., $\widehat{G}$ $=$ $G$ $\cup $ $G^{-1},$ where $G^{-1}$ is the
shadow of $G.$ Futhermore, we are not interested in the disjoint unioned
graphs, because the disjoint union of graphs generates the ``disconnected''
graphs. 

Now, let $G_{k}$ be connected locally finite graphs, and let $v_{k}$ $\in $ $%
V(G_{k})$ be the fixed vertices, for $k$ $=$ $1,$ $2.$ Then, by identifying
the chosen vertices $v_{1}$ and $v_{2},$ we can create a new graph $G$,
denoted by 

\begin{center}
$G_{1}$ $^{v_{1}}\#^{v_{2}}$ $G_{2}.$ 
\end{center}

The identified vertex of $v_{1}$ and $v_{2}$ is called the \emph{glued
vertex of }$v_{1}$\emph{\ }$\in $\emph{\ }$V(G_{1})$\emph{\ and} $v_{2}$ $%
\in $ $V(G_{2}).$ Denote it by $v_{\times }.$ Then the graph $G$ $=$ $G_{1}$ 
$^{v_{1}}\#^{v_{2}}$ $G_{2}$ is the graph with

\begin{center}
$
\begin{array}{ll}
V(G)= & \{v_{\times }\}\cup \left( V(G_{1})\setminus \{v_{1}\}\right)  \\ 
& \,\,\,\,\,\,\,\,\,\,\,\,\,\,\,\cup \left( V(G_{2})\setminus
\{v_{2}\}\right) ,
\end{array}
$
\end{center}

and

\begin{center}
$E(G)$ $=$ $E(G_{1})$ $\cup $ $E(G_{2}),$
\end{center}

under the identification rule; if $e$ $\in $ $E(G_{k})$ satisfies either $e$ 
$=$ $v_{k}$ $e$ or $e$ $=$ $v_{k}$ $e,$ then this edge $e$ is identified
with the edge, also denoted by $e,$ satisfying that $e$ $=$ $v_{\times }$ $e,
$ respectively, $e$ $=$ $e$ $v_{\times }.$ 

This new graph $G$ $=$ $G_{1}$ $^{v_{1}}\#^{v_{2}}$ $G_{2},$ by identifying
the vertices $v_{1}$ $\in $ $V(G_{1})$ and $v_{2}$ $\in $ $V(G_{2}),$ is
called the \emph{glued graph of }$G_{1}$\emph{\ and }$G_{2},$\emph{\ with
the glued vertex of }$v_{1}$\emph{\ and} $v_{2}.$ Again, let $v_{k}$ $\in $ $%
V(G_{k})$ be the fixed vertices, for $k$ $=$ $1,$ $2.$ Define the new
connected locally finite graphs 

\begin{center}
$G_{1}$ $\#^{v_{2}}$ $G_{2}$ $\overset{def}{=}$ $\underset{v\in V(G_{1})}{%
\cup }\left( G_{1}\text{ }^{v}\#^{v_{2}}\text{ }G_{2}\right) $
\end{center}

and

\begin{center}
$G_{1}$ $^{v_{1}}\#$ $G_{2}$ $\overset{def}{=}$ $\underset{v\in V(G_{2})}{%
\cup }\left( G_{1}\text{ }^{v_{1}}\#^{v}\text{ }G_{2}\right) .$
\end{center}

Then the graph $G_{1}$ $\#^{v_{2}}$ $G_{2}$ (resp., $G_{1}$ $^{v_{1}}\#$ $%
G_{2}$) is called the \emph{iterated glued graph with the fixed vertex} $%
v_{2}$ $\in $ $V(G_{2})$ (resp., $v_{1}$ $\in $ $V(G_{1})$). In [16], we
showed that:

\begin{theorem}
(Also, see [16]) Let $G$ be a fractal graph and let $O_{n}$ be the
one-vertex-$n$-loop-edge graph, for $n$ $\in $ $\Bbb{N}.$ If $v$ is the
unique vertex of $O_{n},$ then the iterated glued graph $G$ $\#^{v}$ $O_{n}$
is a fractal graph, too, for all $n$ $\in $ $\Bbb{N}.$
\end{theorem}

\begin{proof}
Roughly speaking the iterated glued graph $G$ $\#^{v}$ $O_{n}$ is the graph
gotten by gluing the unique vertex $v$ of $O_{n}$ to every vertex of $G,$
recursively. Assume that $G$ is a fractal graph, and assume that

\begin{center}
$\deg _{out}(v)$ $=$ $N$ $=$ $\deg _{in}(v),$ in $G,$
\end{center}

for all $v$ $\in $ $V(G),$ and for some $N$ $\in $ $\Bbb{N}.$ Notice that,
by the construction of the iterated glued graph $G$ $\#^{v}$ $O_{n},$

\begin{center}
$V(G)$ $=$ $V(G$ $\#^{v}$ $O_{n}),$
\end{center}

since $O_{n}$ has only one vertex $v,$ for all $n$ $\in $ $\Bbb{N}.$ So, we
can check that

\begin{center}
$\deg _{out}(v)$ $=$ $N$ $+$ $n$ $=$ $\deg _{in}(v),$ in $G$ $\#^{v}$ $O_{n},
$
\end{center}

for all $v$ $\in $ $V(G$ $\#^{v}$ $O_{n})$ $=$ $V(G),$ for all $n$ $\in $ $%
\Bbb{N}.$ Thus, by the graph-theoretical characterization of graph
fractaloids, the graph $G$ $\#^{v}$ $O_{n}$ is a fractal graph, too.
\end{proof}

\strut \strut Again, by the previous example, we can conclude that $O_{n}$ $%
\#^{v}$ $O_{n}$ $\overset{\text{Graph}}{=}$ $O_{2n},$ $K_{n}$ $\#^{v}$ $%
O_{n},$ $C_{n}$ $\#^{v}$ $O_{n},$ and $L$ $\#^{v}$ $O_{n}$ are fractal
graphs, too. For example, $L$ $\#^{v}$ $O_{1}$ is a graph,

\begin{center}
$
\begin{array}{lllllll}
\cdot \cdot \cdot \longrightarrow  & \bullet  & \longrightarrow  & \bullet 
& \longrightarrow  & \bullet  & \longrightarrow \cdot \cdot \cdot  \\ 
& \circlearrowright  &  & \circlearrowright  &  & \circlearrowright  & 
\end{array}
$,
\end{center}

and hence

\begin{center}
$\deg _{out}(v)$ $=$ $2$ $=$ $\deg _{in}(v),$ in $L$ $\#^{v}$ $O_{1},$
\end{center}

for all $v$ $\in $ $V(L$ $\#^{v}$ $O_{1})$ $=$ $V(L).$ Thus $L$ $\#^{v}$ $%
O_{1}$ is a fractal graph.

\section{Radial Operators of Graph Fractaloids\strut \strut}

In this section, we provide a tool to study the property of graph
fractaloids, in operator theory. Let $\Bbb{G}$ be a graph fractaloids,
induced by a connected locally finite directed graph $G.$ Then, as in [19],
we define a suitable Hilbert space operator $T_{G}$, induced by $\Bbb{G},$
and we observe the spectral property of $T_{G},$ by considering the
operator-valued free distributional data of $T_{G}.$ Then such spectral
information of $T_{G}$ explains how the fractality of $\Bbb{G}$ acts on
Hilbert spaces. In [16] and [19], we define $T_{G}$, in terms of the
labelings on the given fractaloid $\Bbb{G},$ on the graph Hilbert space $%
H_{G}.$ Recall that the labelings are determined by the graph automaton $%
\mathcal{A}_{G},$ induced by the graph $G,$ and we called $T_{G},$ the \emph{%
labeling operator of }$\Bbb{G}$\emph{\ on the graph Hilbert space} $H_{G}.$
In this paper, we define graph fractaloids without using automata theory.
Hence, we need to define $T_{G},$ differently.

Let $\Gamma $ be an arbitrary group and let $L(\Gamma )$ $=$ $\overline{\Bbb{%
C}[\theta (\Gamma )]}^{w}$ be the \emph{group von Neumann algebra},
generated by the group $\Gamma ,$ acting on the \emph{group Hilbert space} $%
H_{\Gamma }$ $=$ $l^{2}(\Gamma ),$ where $(H_{\Gamma },$ $\theta )$ is the (%
\emph{left}) \emph{unitary regular representation of }$\Gamma .$ i.e., $%
\theta $ $:$ $\Gamma $ $\rightarrow $ $B(H_{\Gamma })$ is a \emph{group
action}, acting on $H_{\Gamma },$

\begin{center}
$\theta (g)$ $\xi _{g^{\prime }}$ $\overset{def}{=}$ $u_{g}$ $\xi
_{g^{\prime }}$ $=$ $\xi _{gg^{\prime }},$ for all $\xi _{g^{\prime }}$ $\in 
$ $H_{\Gamma },$
\end{center}

where $u_{g}$ is the unitary with its adjoint $u_{g}^{*}$ $=$ $u_{g^{-1}},$
where $g^{-1}$ is the group inverse of $g$ in $\Gamma .$ Clearly, we can
define the ``right'' group action $\pi $ $:$ $\Gamma $ $\rightarrow $ $%
B(H_{\Gamma })$ by

\begin{center}
$\pi (g)$ $\xi _{g^{\prime }}$ $=$ $u_{g}^{op}$ $\xi _{g^{\prime }}$ $%
\overset{def}{=}$ $\xi _{g^{\prime }g},$ for all $g^{\prime }$ $\in $ $%
\Gamma .$
\end{center}

Then the\emph{\ right group von Neumann algebra }$R(\Gamma )$ $=$ $\overline{%
\Bbb{C}[\pi (\Gamma )]}^{w}$ is well-defined in $B(H_{\Gamma }),$ and it is
the opposite $W^{*}$-algebra $L(\Gamma )^{op}$ of $L(\Gamma ),$ and hence
they are anti-$*$-isomorphic from each other. Now, fix a group $\Gamma $ and
its right (or left) group von Neumann algebra $R(\Gamma )$ (resp., $L(\Gamma
)$). Then the \emph{radial operator }$T_{\Gamma }$\emph{\ of} $\Gamma $ in $%
R(\Gamma )$ (resp., $L(\Gamma )$) is defined by

\begin{center}
$T_{\Gamma }$ $\overset{def}{=}$ $\underset{g\in X}{\sum }$ $\pi (g)$ $=$ $%
\underset{g\in \Gamma }{\sum }$ $u_{g}^{op}$
\end{center}

(resp., $T_{\Gamma }$ $\overset{def}{=}$ $\underset{g\in \Gamma }{\sum }$ $%
\theta (g)$ $=$ $\underset{g\in \Gamma }{\sum }$ $u_{g}$) (See [9]), where $%
X $ $\subset $ $\Gamma $ is the set of all generators of $\Gamma .$

Instead of using automata theory, like the radial operators of groups, we
will re-define the labeling operators of graph fractaloids (in the sense of
[20] and [19]) without the labeling process, and we call them, the radial
operators of graph fractaloids.

\subsection{Radial Operators\strut}

Let $G$ be a connected locally finite directed graph with its graph groupoid 
$\Bbb{G},$ and let\strut

\begin{center}
$N$ $\overset{def}{=}$ $\max \{\deg _{out}(v)$ $:$ $v$ $\in $ $V(G)\}.$\strut
\end{center}

Let $M_{G}$ be the right graph von Neumann algebra of $G,$ which is the
groupoid $W^{*}$-algebra $\overline{\Bbb{C}[R(\Bbb{G})]}^{w}$ in $B(H_{G}),$
where $(H_{G},$ $R)$ is the canonical right representation of $\Bbb{G}$ (See
Section 2.4)\strut .

\begin{definition}
Let $M_{G}$ be the right graph von Neumann algebra of $G.$ The radial
operator $T_{G}$ of the graph groupoid $\Bbb{G}$ of $G$ is defined by an
element

\begin{center}
$T_{G}$ $=$ $\underset{e\in E(\widehat{G})}{\sum }$ $R_{e}$ $=$ $\underset{%
e\in E(G)}{\sum }(R_{e}$ $+$ $R_{e^{-1}}),$
\end{center}

in $M_{G},$ where $\widehat{G}$ is the shadowed graph of $G.$
\end{definition}

\strut The radial operator $T_{G}$ of $\Bbb{G}$ is similarly defined like
the Hecke-type operators or the Ruelle operators or the radial operators of
groups. It is easy to check that:

\begin{lemma}
The radial operator $T_{G}$ of the graph groupoid $\Bbb{G}$ of $G$ is
self-adjoint. $\square $
\end{lemma}

Notice that the radial operator $T_{G}$ of $\Bbb{G}$ is defined by the edges
in the ``shadowed'' graph $\widehat{G}$ of $G.$ Therefore, for any summand $%
R_{e},$ we can find its adjoint $R_{e}^{*}$ $=$ $R_{e^{-1}},$ as a summand
of $T_{G}.$ i.e., we can re-define $T_{G}$ by

\begin{center}
$T_{G}$ $=$ $\underset{e\in E(G)}{\sum }\left( R_{e}+R_{e}^{*}\right) .$
\end{center}

Therefore, indeed, the operator $T_{G}$ is self-adjoint in $M_{G}$. By the
self-adjointness of $T_{G},$ the free distributional data of $T_{G},$ in $%
M_{G},$ represents the spectral property of $T_{G}$ on $H_{G}.$

\subsection{Spectral Property of Graph Fractaloids\strut}

By Voiculescu, the spectral property of a ``self-adjoint'' operator $x$ in a
von Neumann algebra $\mathcal{M}$ is represented by the free distributional
data of $x$ over a $W^{*}$-subalgebra $\mathcal{N}$ of $\mathcal{M},$ if
there is a suitable conditional expectation $E$ $:$ $\mathcal{M}$ $%
\rightarrow $ $\mathcal{N}$ (See [5] and [28]). So, to consider the spectral
property of our radial operator $T_{G}$ of $\Bbb{G}$ (on $H_{G}$), we can
compute the free moments $\{E(T_{G}^{n})\}_{n=1}^{\infty },$ where $E$ $:$ $%
M_{G}$ $\rightarrow $ $D_{G}$ is the canonical conditional expectation in
the sense of [10] and [11], where

\begin{center}
$D_{G}$ $\overset{def}{=}$ $\underset{v\in V(\widehat{G})}{\oplus }$ $\left( 
\Bbb{C}\cdot R_{v}\right) $
\end{center}

is the diagonal subalgebra of $M_{G}.$ Since every element $a$ $\in $ $M_{G}$
has its expression

\begin{center}
$a$ $=$ $\underset{w\in \Bbb{G}}{\sum }$ $t_{w}$ $R_{w},$ with $t_{w}$ $\in $
$\Bbb{C},$
\end{center}

we can define $E$ by

\begin{center}
$E(a)$ $\overset{def}{=}$ $\underset{v\in V(\widehat{G})}{\sum }$ $t_{v}$ $%
R_{v},$
\end{center}

for all $a$ $\in $ $M_{G}.$ Then the pair $(M_{G},$ $E)$ is a $D_{G}$-valued 
$W^{*}$-probability space in the sense of Voiculescu. The $D_{G}$-valued
moments of $a$ is defined by the sequence $\{E(a^{n})\}_{n=1}^{\infty },$
and the $D_{G}$-valued $*$-moments of $a$ is defined by the collection,

\begin{center}
$\{E(a^{i_{1}}$ $a^{i_{2}}$ ... $a^{i_{n}})$ $:$ $(i_{1},$ ..., $i_{n})$ $%
\in $ $\{1,$ $*\}^{n},$ $\forall n$ $\in $ $\Bbb{N}\}.$
\end{center}

Then these $D_{G}$-valued $*$-moments of $a$ contain the free distributional
data of $a.$ Clearly, if $a$ is self-adjoint, then the sequence $%
\{E(a^{n})\}_{n=1}^{\infty }$ contains the same data (See [5] and [28]).

\strut Let $\Bbb{R}^{2}$ be the $2$-dimensional $\Bbb{R}$-vector space.
Without loss of generality, we assume $\Bbb{R}^{2}$ is the $\Bbb{R}$-plane,
generated by the horizontal axis and the vertical axis. When we denote the
point $P$ in $\Bbb{R}^{2},$ we will use the pair notation $(a,$ $b),$ as
usual. When we regard the point $P$ in $\Bbb{R}^{2}$ as a vector connecting
the origin $(0,$ $0)$ to $P,$ we use the vector notation $\overrightarrow{(a,%
\text{ }b)}.$ Fix $N$ $\in $ $\Bbb{N}$, and define the \emph{lattices} $%
l_{1},$ ..., $l_{N},$ embedded in $\Bbb{R}^{2},$ by the vectors,

\begin{center}
$l_{k}$ $\overset{def}{=}$ $\overrightarrow{(1,\text{ }e^{k})},$ for all $k$ 
$=$ $1,$ ..., $N,$
\end{center}

where $e$ is the natural exponential number in $\Bbb{R}.$ Define the
corresponding downward lattices $l_{-1},$ $l_{-2},$ ..., $l_{-N}$ by

\begin{center}
$l_{-k}$ $\overset{def}{=}$ $\overrightarrow{(1,\text{ }-e^{k})},$ for all $%
k $ $=$ $1,$ ..., $N.$
\end{center}

Now, construct \emph{lattice paths}, by attaching $l_{\pm 1},$ ..., $l_{\pm
N}$ as follows; $l_{i}$ $l_{j}$ is a lattice path, by transforming the
starting point of the lattice $l_{j}$ to the end point of $l_{i}.$

i.e., we identify the starting point $(0,$ $0)$ of $l_{j}$ to the ending
point $(1,$ $\varepsilon _{i}e^{i})$ of $l_{i},$ where

\begin{center}
$\varepsilon _{i}$ $=$ $\left\{ 
\begin{array}{lll}
1 &  & \text{if }i>0 \\ 
-1 &  & \text{if }i<0,
\end{array}
\right. $
\end{center}

for all $i$ $\in $ $\{\pm 1,$ ..., $\pm N\}.$ Inductively, we can determine
the lattice paths

\begin{center}
$l_{i_{1}}$ $l_{i_{2}}$ ... $l_{i_{n}},$
\end{center}

for all $(i_{1},$ ..., $i_{n})$ $\in $ $\{\pm 1,$ ..., $\pm N\}^{n},$ for
all $n$ $\in $ $\Bbb{N}.$

By $\mathcal{L}_{N},$ we denote the collection of all lattice paths induced
by the lattices $l_{\pm 1},$ ..., $l_{\pm N}.$ Now, let $l_{i_{1}}$ ... $%
l_{i_{n}}$ $\in $ $\mathcal{L}_{N}.$ Then we define the length $\left|
l_{i_{1}}\text{ ... }l_{i_{n}}\right| $ of the given lattice path by the
cardinality $n$ of the lattices, generating the lattice path. i.e.,

\begin{center}
$\left| l_{i_{1}}\text{ ... }l_{i_{n}}\right| $ $=$ $n.$
\end{center}

Define the subset $\mathcal{L}_{N}(n)$ of $\mathcal{L}_{N}$ by the
collection of all length-$n$ lattice paths,

\begin{center}
$\mathcal{L}_{N}(n)$ $\overset{def}{=}$ $\{l$ $\in $ $\mathcal{L}_{N}$ $:$ $%
\left| l\right| $ $=$ $n\},$ for all $n$ $\in $ $\Bbb{N}.$
\end{center}

Then $\mathcal{L}_{N}$ is decomposed by $\mathcal{L}_{N}(n)$'s:

\begin{center}
$\mathcal{L}_{N}$ $=$ $\sqcup _{n=1}^{\infty }$ $\left( \mathcal{L}%
_{N}(n)\right) .$
\end{center}

Let $l$ $\in $ $\mathcal{L}_{N}(n),$ for some $n$ $\in $ $\Bbb{N}.$ And
assume that $l$ starts at $(0,$ $0),$ and $l$ end on the horizontal axis.
Such a lattice path $l$ is said to be a \emph{lattice path} \emph{satisfying
the axis property}. Define the subset $\mathcal{L}_{N}^{o}(n)$ of $\mathcal{L%
}_{N}(n)$ by

\begin{center}
$\mathcal{L}_{N}^{o}(n)$ $\overset{def}{=}$ $\{l$ $\in $ $\mathcal{L}_{N}(n)$
$:$ $l$ satisfies the axis property$\},$
\end{center}

for all $n$ $\in $ $\Bbb{N}.$ By definition, $\mathcal{L}_{N}^{o}(n)$ is
empty, whenever $n$ is odd.

In [19], we found the following $D_{G}$-valued moments of the radial
operator $T_{G}$ of $\Bbb{G},$ where $\Bbb{G}$ is a graph fractaloid.

\begin{theorem}
Let $T_{G}$ be the radial operator of a graph fractaloid $\Bbb{G}$, induced
by a connected locally finite directed graph $G,$ and let

\begin{center}
$N$ $\overset{def}{=}$ $\max \{\deg _{out}(v)$ $:$ $v$ $\in $ $V(G)\}.$
\end{center}

Then

\begin{center}
$E\left( T_{G}^{n}\right) $ $=$ $\left| \mathcal{L}_{N}^{o}(n)\right| \cdot
1_{D_{G}},$ for all $n$ $\in $ $\Bbb{N}.$
\end{center}

More precisely,

\begin{center}
$E(T_{G}^{n})$ $=$ $\left\{ 
\begin{array}{ll}
\left| \mathcal{L}_{N}^{o}(n)\right| \cdot 1_{D_{G}} & \text{if }n\text{ is
even} \\ 
0_{D_{G}} & \text{if }n\text{ is odd.}
\end{array}
\right. $
\end{center}

$\square $\strut 
\end{theorem}

In [40], the first author and his undergraduate students computed the
cardinalities $\left| \mathcal{L}_{N}^{o}(n)\right| $ of $\mathcal{L}%
_{N}^{o}(n)$. In particular, they could find:

\begin{example}
(See [40]) For any $n$ $\in $ $\Bbb{N},$

\begin{center}
$\left| \mathcal{L}_{1}(2n)\right| $ $=$ $_{2n}C_{n},$
\end{center}

and

\begin{center}
$\left| \mathcal{L}_{2}(2n)\right| $ $=$ $_{2n}C_{n}\left( \sum_{j=0}^{n}%
\text{ }\left( _{n}C_{j}\right) ^{2}\right) $,
\end{center}

where $_{m}C_{k}$ $\overset{def}{=}$ $\frac{m!}{n!(m-n)!},$ for all $n$ $%
\leq $ $m$ in $\Bbb{N}.$
\end{example}

\strut More precisely, we have the following computation.

\begin{proposition}
(See [19] and [40]) Let $N$ $\in $ $\Bbb{N},$ and $\mathcal{L}_{N}(n)$, the
collection of all length-$n$ lattice paths induced by the lattices $l_{\pm
1},$ ..., $l_{\pm N}.$ If $\mathcal{L}_{N}^{o}(2n)$ is the subset of $%
\mathcal{L}_{N}(2n),$ consisting of all length-$2n$ lattice paths satisfying
the axis property, for all $n$ $\in $ $\Bbb{N},$ then

\begin{center}
$\left| \mathcal{L}_{N}^{o}(2n)\right| $ $=$ $\underset{(j_{1},...,j_{2n})%
\in \mathcal{C}_{2n}}{\sum }$ $c_{j_{1},...,j_{2n}}$,
\end{center}

where

\begin{center}
$\mathcal{C}_{2n}$ $\overset{def}{=}$ $\left\{ (j_{1},\text{ ..., }%
j_{2n})\left| 
\begin{array}{c}
(j_{1},...,j_{2n})\in \{\pm 1,...,\pm N\}^{2n}, \\ 
j_{1}\leq j_{2}\leq ...\leq j_{2n} \\ 
\sum_{k=1}^{2n}j_{k}=0
\end{array}
\right. \right\} ,$
\end{center}

and where the numbers $c_{j_{1}...j_{2n}}$ satisfies the recurrence relation:

\begin{center}
$c_{j_{1}...j_{2n}}$ $=$ $c_{j_{1},...,j_{2n-m},\text{ }\underset{m\text{%
-times}}{\underbrace{j_{0},\text{ }j_{0},.......,,j_{0},\text{ }j_{0}}}}$ $=$
$c_{j_{1},...,j_{2n-m}}$ $\cdot $ $_{2n}C_{m},$
\end{center}

with

\begin{center}
$c_{j,\text{ }j,\text{ ..., }j}$ $=$ $1,$ for all $j$ $\in $ $\{\pm 1,$ ..., 
$\pm N\},$
\end{center}

for all $1$ $\leq $ $m$ $\leq $ $2n,$ in $\Bbb{N}$. Here, $_{n}C_{k}$ means $%
\frac{n!}{k!(n-k)!},$ for all $k$ $\leq $ $n$ $\in $ $\Bbb{N}.$ $\square $
\end{proposition}

\strut For instance, assume that we have $c_{-3,\,-2,-2,-1,\text{ }1,\text{ }%
2,\text{ }2,\,3}$, as a summand of $\left| \mathcal{L}_{N}^{o}(6)\right| ,$
for $N$ $\geq $ $3.$ Indeed, it is a summand of $\left| \mathcal{L}%
_{k}^{o}(6)\right| ,$ since

\begin{center}
$(-3,$ $-2,$ $-2,$ $-1,$ $1,$ $2,$ $2,$ $3)$ $\in $ $\mathcal{C}_{6}.$
\end{center}

Then, by the recurrence relation, we can compute it by

\begin{center}
$
\begin{array}{ll}
c_{-3,-2,-2,-1,1,2,2,3} & =c_{-3,-2,-2,-1,1,2,2}\cdot _{6}C_{1} \\ 
& =c_{-3,-2,-2,-1}\text{ }\cdot \text{ }_{5}C_{2}\cdot \text{ }_{6}C_{1} \\ 
& =c_{-3,-2,-2\text{ }}\cdot \text{ }_{4}C_{1}\text{ }\cdot \text{ }_{5}C_{2}%
\text{ }\cdot \text{ }_{6}C_{1} \\ 
& =c_{-3}\text{ }\cdot \text{ }_{3}C_{2}\text{ }\cdot \text{ }_{4}C_{1}\text{
}\cdot \text{ }_{5}C_{2}\text{ }\cdot \text{ }_{6}C_{1} \\ 
& =1\cdot \text{ }_{3}C_{2}\text{ }\cdot \text{ }_{4}C_{1}\text{ }\cdot 
\text{ }_{5}C_{2}\text{ }\cdot \text{ }_{6}C_{1}.
\end{array}
$
\end{center}

\section{\strut Classification of Graph Fractaloids}

Let $A_{1}$ and $A_{2}$ be von Neumann algebras and let $a_{k}$ $\in $ $%
A_{k} $ be arbitrary fixed ``self-adjoint'' operators, for $k$ $=$ $1,$ $2.$
Assume that $A_{1}$ and $A_{2}$ contains their $W^{*}$-subalgebras $B_{1}$
and $B_{2},$ which are $*$-isomorphic to a von Neumann algebra $B.$ Without
loss of generality, assume that $A_{1}$ and $A_{2}$ contain their common $%
W^{*}$-subalgebra $B.$ Let $E_{k}$ $:$ $A_{k}$ $\rightarrow $ $B$ be a
conditional expectation, for $k$ $=$ $1,$ $2.$ We say that the \emph{%
self-adjoint operators }$a_{1}$\emph{\ and }$a_{2}$\emph{\ are} \emph{%
identically free distributed over} $B,$ if

\begin{center}
$E(a_{1}^{n})$ $=$ $E(a_{2}^{n})$ in $B,$ for all $n$ $\in $ $\Bbb{N}.$
\end{center}

If two connected locally finite directed graphs $G_{1}$ and $G_{2}$ are
fractal graphs, and if

\begin{center}
$\deg _{out}(v_{1})$ $=$ $N$ $=$ $\deg _{out}(v_{2}),$
\end{center}

for $v_{1}$ $\in $ $V(G_{1})$ and $v_{2}$ $\in $ $V(G_{2}),$ then the
corresponding graph fractaloids $\Bbb{G}_{1}$ and $\Bbb{G}_{2}$ have the
identically free distributed spectral properties ``up to identity
elements.'' Indeed, if $N$ $\in $ $\Bbb{N}$ is given as above, then

\begin{center}
\strut $E(T_{G_{1}}^{n})$ $=$ $\left| \mathcal{L}_{N}^{o}(n)\right| \cdot $ $%
1_{D_{G_{1}}}$
\end{center}

and

\begin{center}
$E(T_{G_{2}}^{n})$ $=$ $\left| \mathcal{L}_{N}^{o}(n)\right| $ $\cdot $ $%
1_{D_{G_{2}}},$
\end{center}

where $\mathcal{L}_{N}^{o}(n)$ is the collection of all length-$n$ lattice
paths, induced by the lattices $l_{\pm 1},$ ..., $l_{\pm N},$ satisfying the
axis property, for $n$ $\in $ $\Bbb{N}.$ For instance, let $G_{1}$ $=$ $%
K_{n} $ be the one-flow circulant graph with $n$-vertices, for $n$ $\in $ $%
\Bbb{N}, $ and $G_{2}$ $=$ $L,$ where $L$ is the infinite linear graph with

\begin{center}
$V(L)$ $=$ $\{$...$,$ $v_{-2},$ $v_{-1},$ $v_{0},$ $v_{1},$ $v_{2},$ $v_{3},$
...$\}$
\end{center}

and

\begin{center}
$E(L)$ $=$ $\{e_{j}$ $=$ $v_{j}$ $e_{j}$ $v_{j+1}$ $:$ $j$ $\in $ $\Bbb{Z}%
\}. $
\end{center}

Then the corresponding graph groupoids $\Bbb{G}_{1}$ and $\Bbb{G}_{2}$ are
graph fractaloids, since

\begin{center}
$\deg _{out}^{(G_{1})}(v)$ $=$ $1$ $=$ $\deg _{in}^{(G_{1})}(v),$ in $G_{1},$
\end{center}

and

\begin{center}
$\deg _{out}^{(G_{2})}(x)$ $=$ $1$ $=$ $\deg _{in}^{(G_{2})}(x)$, in $G_{2},$
\end{center}

for all $v$ $\in $ $V(G_{1}),$ and $x$ $\in $ $V(G_{2}).$ So, the radial
operators $T_{G_{1}}$ and $T_{G_{2}}$ have their amalgamated free moments,

\begin{center}
$E\left( T_{G_{1}}^{k}\right) $ $=$ $\left| \mathcal{L}_{1}^{o}(k)\right| $ $%
\cdot $ $1_{\Bbb{C}^{\oplus \,n}},$
\end{center}

and

\begin{center}
$E\left( T_{G_{2}}^{k}\right) $ $=$ $\left| \mathcal{L}_{1}^{o}(k)\right| $ $%
\cdot $ $1_{\Bbb{C}^{\oplus \,\infty }},$
\end{center}

because

\begin{center}
$D_{G_{1}}$ $\overset{*\text{-isomorphic}}{=}$ $\Bbb{C}^{\oplus \,n},$ in $%
M_{G_{1}},$
\end{center}

and

\begin{center}
$D_{G_{1}}$ $\overset{*\text{-isomorphic}}{=}$ $\Bbb{C}^{\oplus \,\infty },$
in $M_{G_{2}}.$
\end{center}

\subsection{Fractal Pairs of Graph Fractaloids}

Motivated by the example, in the previous paragraph, we can obtain the
following proposition.

\begin{proposition}
Let $G_{1}$ and $G_{2}$ be connected locally finite directed graphs with

\begin{center}
$
\begin{array}{ll}
\max \{\deg _{out}(v_{1}):v_{1}\in V(G_{1})\} & =N_{0} \\ 
& =\max \{\deg _{out}(v_{2}):v_{2}\in V(G_{2})\},
\end{array}
$
\end{center}

in $\Bbb{N},$ and

\begin{center}
$\left| V(G_{1})\right| $ $=$ $N^{0}$ $=$ $\left| V(G_{2})\right| $,
\end{center}

where $N_{0}$ $\in $ $\Bbb{N}$ $\cup $ $\{\infty \}.$ If $G_{1}$ and $G_{2}$
are fractal graphs, then the radial operators $T_{G_{1}}$ and $T_{G_{2}}$
are identically free distributed over $\Bbb{C}^{\oplus \,N^{0}}.$
\end{proposition}

\begin{proof}
\strut Since $\left| V(G_{1})\right| $ $=$ $N^{0}$ $=$ $\left|
V(G_{2})\right| $ in $\Bbb{N}$ $\cup $ $\{\infty \},$ the diagonal
subalgebras $D_{G_{1}}$ and $D_{G_{2}}$ are $*$-isomorphic to $\Bbb{C}%
^{\oplus \,N^{0}}.$ Let's denote $\Bbb{C}^{\oplus N^{0}}$ by $B.$ By the
fractality of the graphs $G_{1}$ and $G_{2}$, we have

\begin{center}
$E(T_{G_{1}}^{n})$ $=$ $\left| \mathcal{L}_{N_{0}}^{o}(n)\right| $ $\cdot $ $%
1_{B}$ $=$ $E(T_{G_{2}}^{n}),$
\end{center}

for all $n$ $\in $ $\Bbb{N}$, where $1_{B}$ is the identity $(N^{0}$ $\times 
$ $N^{0})$-matrix in $B.$ Therefore, the radial operators $T_{G_{1}}$ and $%
T_{G_{2}}$ are identically free distributed over $B.$
\end{proof}

\strut The above proposition shows that the pair $(N_{0},$ $N^{0})$ of the
numbers $N_{0}$ $\in $ $\Bbb{N},$ and $N^{0}$ $\in $ $\Bbb{N}$ $\cup $ $%
\{\infty \},$ can explain the fractality (characterized by the spectral
information of the radial operator $T_{G}$) of a graph fractaloid $\Bbb{G},$
via the sequence

\begin{center}
$\left( \left| \mathcal{L}_{N_{0}}^{o}(n)\right| \right) _{n=1}^{\infty }$
over $\Bbb{C}^{\oplus \,N^{0}},$
\end{center}

where

\begin{center}
$N_{0}$ $=$ $\deg _{out}(v)$ $=$ $\deg _{in}(v),$ in $G,$
\end{center}

for all $v$ $\in $ $V(G),$ and

\begin{center}
$N^{0}$ $=$ $\left| V(G)\right| .$
\end{center}

\begin{definition}
\strut Let $G$ be a connected locally finite directed graph with its graph
groupoid $\Bbb{G},$ and assume that $\Bbb{G}$ is a graph fractaloid. Then
the pair $(N_{0},$ $N^{0})$, where

\begin{center}
$N_{0}$ $=$ $\max \{\deg _{out}(v)$ $:$ $v$ $\in $ $V(G)\}$ $\in $ $\Bbb{N}$
\end{center}

and

\begin{center}
$N^{0}$ $=$ $\left| V(G)\right| $ $\in $ $\Bbb{N}$ $\cup $ $\{\infty \}$ $%
\overset{denote}{=}$ $\Bbb{N}_{\infty },$
\end{center}

is called the fractal pair of $\Bbb{G}.$ Denote the fractal pair $(N_{0},$ $%
N^{0})$ of a graph fractaloid $\Bbb{G}$ by $fp(\Bbb{G}).$
\end{definition}

By the previous proposition, we can get the following theorem.

\begin{theorem}
Let $G_{k}$ be connected locally finite fractal graphs with their graph
fractaloids $\Bbb{G}_{k},$ for $k$ $=$ $1,$ $2.$ If the fractal pairs of $%
\Bbb{G}_{k}$ are identical to $(N_{0},$ $N^{0}),$ for $N_{0}$ $\in $ $\Bbb{N}%
,$ and $N^{0}$ $\in $ $\Bbb{N}_{\infty },$ then the radial operators $%
T_{G_{k}}$ of $\Bbb{G}_{k}$ are identically free distributed over $\Bbb{C}%
^{\oplus \,N^{0}}.$ Moreover,

\begin{center}
$E(T_{G_{k}}^{n})$ $=$ $\left| \mathcal{L}_{N_{0}}^{o}(n)\right| \cdot 1_{%
\Bbb{C}^{\oplus N^{0}}},$
\end{center}

for all $n$ $\in $ $\Bbb{N}.$ $\square $
\end{theorem}

The fractal pairs of graph fractaloids give the classification of graph
fractaloids in terms of the spectral information of the corresponding radial
operators.

\subsection{Equivalence Classes of Graph Fractaloids}

Now, let's collect all graph fractaloids induced by ``connected'' fractal
graphs, and denote this collection by $\mathcal{F}_{ractal},$ i.e.,

\begin{center}
$\mathcal{F}_{ractal}$ $\overset{def}{=}$ $\left\{ \Bbb{G}\left| 
\begin{array}{c}
\Bbb{G}\text{ is a graph fractaloid of }G, \\ 
\text{where }G\text{ is a connected} \\ 
\text{locally finite fractal graph.}
\end{array}
\right. \right\} .$
\end{center}

\strut Define an equivalence relation $\mathcal{R}$ on the set $\mathcal{F}%
_{ractal}$ by

\begin{center}
$\Bbb{G}_{1}$ $\mathcal{R}$ $\Bbb{G}_{2}$ $\overset{def}{\Longleftrightarrow 
}$ $fp(\Bbb{G}_{1})$ $=$ $fp(\Bbb{G}_{2})$ in $\Bbb{N}$ $\times $ $\Bbb{N}%
_{\infty }.$
\end{center}

Then the relation $\mathcal{R}$ on $\mathcal{F}_{ractal}$ is indeed an
equivalence relation:

(5.1)\ \quad $\Bbb{G}\mathcal{R}\Bbb{G},$ for all $\Bbb{G}$ $\in $ $\mathcal{%
F}_{ractal},$

(5.2) $\quad \Bbb{G}_{1}\mathcal{R}\Bbb{G}_{2}$ $\Longrightarrow $ $\Bbb{G}%
_{2}\mathcal{R}\Bbb{G}_{1},$ and

(5.3) $\quad \Bbb{G}_{1}\mathcal{R}\Bbb{G}_{2}$ and $\Bbb{G}_{2}\mathcal{R}%
\Bbb{G}_{3}$ $\Longrightarrow $ $\Bbb{G}_{1}\mathcal{R}\Bbb{G}_{3},$

for all $\Bbb{G},$ $\Bbb{G}_{1},$ $\Bbb{G}_{2},$ $\Bbb{G}_{3}$ $\in $ $%
\mathcal{F}_{ractal}.$ By (5.1), (5.2), and (5.3), the relation $\mathcal{R}$
is an equivalence relation on the set $\mathcal{F}_{ractal}.$

\begin{definition}
Let $\mathcal{F}_{ractal}$ be the collection of all graph fractaloids,
induced by connected locally finite directed graphs, and let $\mathcal{R}$
be an equivalence relation on $\mathcal{F}_{ractal},$ defined as above. Then
we call $\mathcal{R}$, the spectral (equivalence) relation on $\mathcal{F}%
_{ractal}.$ And the equivalence classes of $\mathcal{R}$ are said to be
spectral classes of $\mathcal{F}_{ractal.}$ Denote each spectral class by $%
[(N_{0},$ $N^{0})],$ for all $(N_{0},$ $N^{0})$ $\in $ $\Bbb{N}$ $\times $ $%
\Bbb{N}_{\infty }.$ i.e.,

\begin{center}
$[(N_{0},$ $N^{0})]$ $\overset{def}{=}$ $\{\Bbb{G}$ $\in $ $\mathcal{F}%
_{ractal}$ $:$ $fp(\Bbb{G})$ $=$ $(N_{0},$ $N^{0})\}.$
\end{center}
\end{definition}

\strut Now, fix $(N_{0},$ $N^{0})$ $\in $ $\Bbb{N}$ $\times $ $\Bbb{N}%
_{\infty }.$ Then we can always choose at least one element $\Bbb{G}$ in $%
\mathcal{F}_{ractal},$ equivalently, for any $(N_{0},$ $N^{0})$ $\in $ $\Bbb{%
N}$ $\times $ $\Bbb{N}_{\infty },$ we can always have a ``nonempty'' graph
fractaloid $\Bbb{G}$ $\in $ $\mathcal{F}_{ractal}$, such that $fp(\Bbb{G})$ $%
=$ $(N_{0},$ $N^{0}).$ The following lemma is proven by construction.

\begin{lemma}
Let $(N_{0},$ $N^{0})$ $\in $ $\Bbb{N}$ $\times $ $\Bbb{N}_{\infty }$. Then
there exists at least one $\Bbb{G}$ $\in $ $\mathcal{F}_{ractal},$ such that 
$\Bbb{G}$ $\in $ $[(N_{0},$ $N^{0})],$ where $[(N_{0},$ $N^{0})]$ is the
spectral class. In other words, each spectral class is nonempty.
\end{lemma}

\begin{proof}
\strut Fix $(N_{0},$ $1)$ $\in $ $\Bbb{N}$ $\times $ $\{1\}.$ Then we can
construct the one-vertex-$N_{0}$-loop-edge graph $O_{N_{0}},$ generating the
graph fractaloid $\Bbb{O}_{N_{0}},$ which is the graph groupoid of $%
O_{N_{0}}.$ Now, take $(N_{0},$ $N^{0})$ $\in $ $\Bbb{N}$ $\times $ $\Bbb{N}.
$ Then we can construct the $N_{0}$-regularized graph $R_{N_{0}}(K_{N^{0}}),$
in the sense of Section 3.2, of the one-flow circulant graph $K_{N^{0}}$
with $N^{0}$-vertices, generating the graph fractaloid $\Bbb{G}\left(
R_{N_{0}}(K_{N^{0}})\right) $ in $\mathcal{F}_{ractal}.$ Assume that $(N_{0},
$ $\infty )$ $\in $ $\Bbb{N}$ $\times $ $\{\infty \}.$ Construct the
infinite linear graph $L,$ graph-isomorphic to

\begin{center}
$\cdot \cdot \cdot \longrightarrow \bullet \longrightarrow \bullet
\longrightarrow \bullet \longrightarrow \cdot \cdot \cdot $ .
\end{center}

Then the graph groupoid $\Bbb{L}$ of $L$ is in $\mathcal{F}_{ractal},$ with
its fractal pair $fp(\Bbb{L})$ $=$ $(1,$ $\infty ).$ If we construct the $%
N_{0}$-regularized graph $R_{N_{0}}(L),$ then the fractal pair of the graph
fractaloid $\Bbb{G}\left( R_{N_{0}}(L)\right) $ is identical to $(N_{0},$ $%
\infty ).$

Therefore, for any $(N_{0},$ $N^{0})$ $\in $ $\Bbb{N}$ $\times $ $\Bbb{N}%
_{\infty },$ there always exists at least one graph fractaloid $\Bbb{G}$ in $%
\mathcal{F}_{ractal},$ such that $fp(\Bbb{G})$ $=$ $(N_{0},$ $N^{0}).$
\end{proof}

\strut By the previous lemma, we can get the following theorem of this
paper. Actually the following theorem is the summary of all our main results
of this paper.

\begin{theorem}
Let $(N_{0},$ $N^{0})$ $\in $ $\Bbb{N}$ $\times $ $\Bbb{N}_{\infty }.$ Then
the corresponding spectral class $[(N_{0},$ $N^{0})]$ of $\mathcal{F}%
_{ractal}$ is nonempty. Moreover, if $\Bbb{G}$ $\in $ $[(N_{0},$ $N^{0})],$
and if $T_{G}$ is the radial operator of $\Bbb{G}$ in the right graph von
Neumann algebra $M_{G},$ then the $D_{G}$-valued moments $E(T_{G}^{n})$
satisfy

\begin{center}
$E(T_{G}^{n})$ $=$ $\left| \mathcal{L}_{N_{0}}^{o}(n)\right| \cdot 1_{\Bbb{C}%
^{\oplus \,N^{0}}},$ for all $n$ $\in $ $\Bbb{N}.$
\end{center}

$\square $
\end{theorem}

The following corollary is the direct consequnece of the above theorem.

\begin{corollary}
Let $\mathcal{F}_{ractal}$ be the set of all graph fractaloids induced by
connected locally finite directed graphs. Then $\mathcal{F}_{ractal}$ is
classified by the spectral classes $[(n,$ $m)],$ for all $(n,$ $m)$ $\in $ $%
\Bbb{N}$ $\times $ $\Bbb{N}_{\infty }.$ i.e.,

\begin{center}
$\mathcal{F}_{ractal}$ $=$ $\underset{(n,\text{ }m)\in \Bbb{N}\text{ }\times 
\text{ }\Bbb{N}_{\infty }}{\sqcup }$ $\left( [(n,\text{ }m)]\right) ,$
\end{center}

where $\sqcup $ means the disjoint union. $\square $
\end{corollary}

\strut Notice that our spectral classes in $\mathcal{F}_{ractal}$ are
determined by the spectral information of graph fractaloids, not by the
graph-theoretical (and algebraic) information of fractal graphs (resp.,
graph fractaloids).

\begin{proposition}
Let $G$ be a connected locally finite directed graph with its graph groupoid 
$\Bbb{G},$ and assume that $\Bbb{G}$ is a graph fractaloid in $\mathcal{F}%
_{ractal}.$ Let $G^{\prime }$ be a directed graph. If $G$ and $G^{\prime }$
are graph-isomorphic, then the graph groupoid $\Bbb{G}^{\prime }$ of $%
G^{\prime }$ is a graph fractaloid in $\mathcal{F}_{ractal},$ too. Moreover,
the graph fractaloids $\Bbb{G}$ and $\Bbb{G}^{\prime }$ are contained in the
same spectral class $[(n,$ $m)],$ for some $n$ $\in $ $\Bbb{N},$ and $m$ $%
\in $ $\Bbb{N}_{\infty }.$
\end{proposition}

\begin{proof}
\strut Since the graphs $G$ and $G^{\prime }$ are graph-isomorphic, the
graph groupoids $\Bbb{G}$ of $G$ and $\Bbb{G}^{\prime }$ of $G^{\prime }$
are groupoid-isomorphic. And since $\Bbb{G}$ $\in $ $\mathcal{F}_{ractal},$
the graph groupoid $\Bbb{G}^{\prime }$ $\in $ $\mathcal{F}_{ractal},$ too.
Assume now that $fp(\Bbb{G})$ $=$ $(n,$ $m)$ $\in $ $\Bbb{N}$ $\times $ $%
\Bbb{N}_{\infty }.$ Then, since

\begin{center}
$V(G)$ $=$ $V(\widehat{G})$ $=$ $V(\widehat{G^{\prime }})$ $=$ $V(G^{\prime
}),$
\end{center}

we have

\begin{center}
$\left| V(G)\right| $ $=$ $m$ $=$ $\left| V(G^{\prime })\right| .$
\end{center}

Also, since $\widehat{G}$ and $\widehat{G^{\prime }}$ are graph-isomorphic,

\begin{center}
$\max \{\deg _{out}(v)$ $:$ $v$ $\in $ $V(G^{\prime })\}$ $=$ $n,$
\end{center}

too, because

\begin{center}
$n$ $=$ $\max \{\deg _{out}(v)$ $:$ $v$ $\in $ $V(G)\}.$
\end{center}

So, we obtain that

\begin{center}
$fp(\Bbb{G}^{\prime })$ $=$ $(n,$ $m),$
\end{center}

too. Therefore, the graph fractaloids $\Bbb{G}$ and $\Bbb{G}^{\prime }$ are
contained in the same spectral class $[(n,$ $m)]$ in $\mathcal{F}_{ractal}.$
\end{proof}

\strut How about the converse of the previous proposition? The following
example shows that the converse does not hold true.

\begin{example}
Let $K_{3}$ be the one-flow circulant graph with 3-vertices, and let $G_{1}$
be the $2$-regularized graph $R_{2}(K_{3})$ of $K_{3}.$ Then $G_{1}$ is a
fractal graph, and hence the graph groupoid $\Bbb{G}_{1}$ is a graph
fractaloid in $\mathcal{F}_{ractal},$ since

\begin{center}
$\deg _{out}^{(G_{1})}(v)$ $=$ $2$ $=$ $\deg _{in}^{(G_{1})}(v)$, in $G_{1}$ 
$=$ $R_{2}(K_{3}),$
\end{center}

for all $v$ $\in $ $V(G_{1}).$ Now, let $G_{2}$ be the complete graph $C_{3}$
with 3-vertices. Then

$\deg _{out}^{(G_{2})}(x)$ $=$ $2$ $=$ $\deg _{in}^{(G_{2})}(x),$ in $G_{2}$ 
$=$ $C_{3},$

for all $x$ $\in $ $V(G_{2}).$ Thus, the graph groupoid $\Bbb{G}_{2}$ of $%
G_{2}$ is a graph fractaloid in $\mathcal{F}_{ractal},$ too. Moreiver, both $%
\Bbb{G}_{1}$ and $\Bbb{G}_{2}$ have the same fractal pair $(2,$ $3),$ and
hence 

\begin{center}
$\Bbb{G}_{1},$ $\Bbb{G}_{2}$ $\in $ $[(2,$ $3)]$ $\in $ $\mathcal{F}%
_{ractal}.$
\end{center}

But the graphs $R_{2}(K_{3})$ and $C_{3}$ are not graph-isomorphic. This
shows that the converse of the previous proposition does not hold true, in
general.
\end{example}

\strut In the previous example, we showed that even though two graphs are
not graph-isomorphic, the corresponding graph groupoids can be fractaloids
contained in the same spectral class in the set $\mathcal{F}_{ractal}.$ More
general to the previous proposition, we can obtain the following generalized
result.

\begin{theorem}
Let $G_{1}$ and $G_{2}$ be connected locally finite graphs, and assume that
the corresponding shadowed graphs $\widehat{G_{1}}$ and $\widehat{G_{2}}$
are graph-isomorphic. If $G_{1}$ is a fractal graph, then $G_{2}$ is a
fractal graph, too. Moreover, the graph fractaloids $\Bbb{G}_{k}$ of $G_{k}$
are contained in the same spectral class in $\mathcal{F}_{ractal}.$
\end{theorem}

\begin{proof}
\strut Let $G_{1}$ be a fractal graph with its graph fractaloid $\Bbb{G}_{1},
$ and let $\Bbb{G}_{1}$ $\in $ $[(n,$ $m)]$ $\subset $ $\mathcal{F}_{ractal},
$ for some $n$ $\in $ $\Bbb{N},$ and $m$ $\in $ $\Bbb{N}_{\infty }.$ And
assume that the shadowed graphs $\widehat{G_{1}}$ and $\widehat{G_{2}}$ are
graph-isomorphic. Then the graph groupoids $\Bbb{G}_{1}$ and $\Bbb{G}_{2}$
are groupoid-isomorphic. So, the radial operators $T_{G_{k}}$ of $\Bbb{G}_{k}
$ are identically distributed over $\Bbb{C}^{\oplus \,m},$ and hence $fp(%
\Bbb{G}_{2})$ $=$ $(n,$ $m),$ too. Therefore, $\Bbb{G}_{2}$ $\in $ $[(n,$ $%
m)]$ in $\mathcal{F}_{ractal}.$
\end{proof}

\strut How about the converse of the above theorem? Unfortunately, the
converse does not hold, in general, either.

\begin{example}
Let $G_{1}$ be the $2$-regularized graph $R_{2}(K_{3})$, and let $G_{2}$ be
the iterated glued graph $K_{3}$ $\#^{v}$ $O_{1},$ where $v$ is the unique
vertex of the one-vertex-1-loop-edge graph $O_{1}.$ As we observed in
Section 3.2, since the one-flow circulant graph $K_{3}$ is a fractal graph,
its $2$-regularized graph $R_{2}(K_{3}),$ and the iterated glued graph $K_{3}
$ $\#^{v}$ $O_{1}$ are fractal graphs, too. i.e., the graph groupoids $\Bbb{G%
}_{k}$ of $G_{k}$ are graph fractaloids, for $k$ $=$ $1,$ $2.$ Moreover, we
have that

\begin{center}
$\left| V(G_{1})\right| $ $=$ $3$ $=$ $\left| V(K_{3})\right| $ $=$ $\left|
V(G_{2})\right| ,$
\end{center}

and

\begin{center}
$\deg _{out}^{(G_{1})}(v)$ $=$ $2$ $=$ $\deg _{in}^{(G_{1})}(v),$ in $G_{1},$
\end{center}

for all $v$ $\in $ $V(G_{1}),$ and

\begin{center}
$\deg _{in}^{(G_{2})}(x)$ $=$ $2$ $=$ $\deg _{in}^{(G_{2})}(x),$ in $G_{2},$
\end{center}

for all $x$ $\in $ $V(G_{2}).$ i.e., the fractal pairs $fp(\Bbb{G}_{k})$ of
the graph fractaloids $\Bbb{G}_{k}$ are

\begin{center}
$fp(\Bbb{G}_{1})$ $=$ $(2,$ $3)$ $=$ $fp(\Bbb{G}_{2}).$
\end{center}

Therefore, the graph fractaloids $\Bbb{G}_{1}$ and $\Bbb{G}_{2}$ are
contained in the spectral class $[(2,$ $3)]$ in $\mathcal{F}_{ractal}.$
However, it is easy to check that the shadowed graphs $\widehat{G_{1}}$ and $%
\widehat{G_{2}}$ are not graph-isomorphic. Indeed, the shadowed graphs $%
\widehat{G_{2}}$ contain loop edges, but $\widehat{G_{1}}$ does not contain
any loop edges. This example shows that even though the shadowed graphs $%
\widehat{G_{1}}$ and $\widehat{G_{2}}$ are not graph-isomorphic, it is
possible that the graph fractaloids $\Bbb{G}_{1}$ and $\Bbb{G}_{2}$ are
contained in the same spectral class in $\mathcal{F}_{ractal}.$ 
\end{example}

By the previous example, the converse of the previous theorem does not hold
true, in general.\strut 

\strut 

\strut \textbf{References}\strut

\strut

\begin{quote}
{\small [1] \ \ A. G. Myasnikov and V. Shapilrain (editors), Group Theory,
Statistics and Cryptography, Contemporary Math, 360, (2003) AMS.}

{\small [2] \ \ A. Gibbons and L. Novak, Hybrid Graph Theory and Network
Analysis, ISBN: 0-521-46117-0, (1999) Cambridge Univ. Press.}

{\small [3]\strut \ \ \ \strut B. Solel, You can see the arrows in a Quiver
Operator Algebras, (2000), preprint.}

{\small [4] \ \ C. W. Marshall, Applied Graph Theory, ISBN: 0-471-57300-0
(1971) John Wiley \& Sons}

{\small [5] \ \ \strut D.Voiculescu, K. Dykemma and A. Nica, Free Random
Variables, CRM Monograph Series Vol 1 (1992).\strut }

{\small [6] \ \ D.W. Kribs and M.T. Jury, Ideal Structure in Free
Semigroupoid Algebras from Directed Graphs, preprint.}

{\small [7] \ \ D.W. Kribs, Quantum Causal Histories and the Directed Graph
Operator Framework, arXiv:math.OA/0501087v1 (2005), Preprint.}

{\small [8] \ \ I. Cho, Group Freeness and Certain Amalgamated Freeness, J.
of KMS, 45, no. 3, (2008) 597 - 609.}

{\small [9] \ \ I. Cho, The Moments of Certain Perturbed Operators of the
Radial Operator of the Free Group Factor }$L(F_{N})${\small , JAA, 5, no. 3,
(2007) 137 - 165.}

{\small [10]\ I. Cho, Graph von Neumann algebras, ACTA. Appl. Math, 95,
(2007) 95 - 135.}

{\small [11]\ I. Cho, Characterization of Free Blocks of a right graph von
Neumann algebra, CAOT, 1, (2007) 367 - 398.}

{\small [12]\ I. Cho, Operator Algebraic Structures Induced by Graphs, CAOT,
(2009), To Appear.}

{\small [13] I. Cho, Vertex-Compressed Algebras of a Graph von Neumann
Algebra, ACTA Appl. Math., (2009) To Appear.}

{\small [14] I. Cho, Graph Groupoids and Corresponding Representations,
Group Theory: Classes, Representations and Connections, and Applications,
(2009) NOVA Publisher.}

{\small [15] I. Cho, Measures on Graphs and Groupoid Measures, CAOT, 2,
(2008) 1 - 28. }

{\small [16] I. Cho, and P. E. T. Jorgensen, Graph Fractaloids: Graph
Groupoids with Fractal Property, (2008) Submitted to J. of Phy. A.}

{\small [17]\ I. Cho, and P. E. T. Jorgensen, }$C^{*}${\small -Algebras
Generated by Partial Isometries, JAMC, (2009) To Appear.}

{\small [18] I. Cho, and P. E. T. Jorgensen, }$C^{*}${\small -Subalgebras
Generated by Partial Isometries, JMP, (2009) To Appear.}

{\small [19] I. Cho, and P. E. T. Jorgensen, Applications of Automata and
Graphs: Labeling-Operators in Hilbert Space I, ACTA Appl. Math.: Special
Issues (2009) To Appear. }

{\small [20] I. Cho, and P. E. T. Jorgensen, Applications of Automata and
Graphs: Labeling-Operators in Hilbert Space II, (2008) Submitted to JMP.}

{\small [21] I. Cho, and P. E. T. Jorgensen, }$C^{*}${\small -Subalgebras
Generated by Single Operator in }$B(H)${\small , ACTA Appl. Math: Special
Issues, (2009) To Appear.}

{\small [22] I. Cho, and P. E. T. Jorgensen, Measure Framing on Graphs, and
Framed Graph von Neumann Algebras, (2009) Preprint.}

{\small [23]\ I. Raeburn, Graph Algebras, CBMS no 3, AMS (2005).}

{\small [24]\ P. D. Mitchener, }$C^{*}${\small -Categories, Groupoid
Actions, Equivalent KK-Theory, and the Baum-Connes Conjecture,
arXiv:math.KT/0204291v1, (2005), Preprint.}

{\small [25] R. Scapellato and J. Lauri, Topics in Graph Automorphisms and
Reconstruction, London Math. Soc., Student Text 54, (2003) Cambridge Univ.
Press.}

{\small [26] R. Exel, A new Look at the Crossed-Product of a }$C^{*}${\small %
-algebra by a Semigroup of Endomorphisms, (2005) Preprint.}

{\small [27] R. Gliman, V. Shpilrain and A. G. Myasnikov (editors),
Computational and Statistical Group Theory, Contemporary Math, 298, (2001)
AMS.}

{\small [28] R. Speicher, Combinatorial Theory of the Free Product with
Amalgamation and Operator-Valued Free Probability Theory, AMS Mem, Vol 132 ,
Num 627 , (1998).}

{\small [29] S. H. Weintraub, Representation Theory of Finite Groups:
Algebra and Arithmetic, Grad. Studies in Math, vo. 59, (2003) AMS.}

{\small [30] V. Vega, Finite Directed Graphs and }$W^{*}${\small %
-Correspondences, (2007) Ph. D thesis, Univ. of Iowa.}

{\small [31] W. Dicks and E. Ventura, The Group Fixed by a Family of
Injective Endomorphisms of a Free Group, Contemp. Math 195, AMS.}

{\small [32] D. A. Lind, Entropies of Automorphisms of a Topological Markov
Shift, Proc. AMS, vo 99, no 3, (1987) 589 - 595.}

{\small [33] D. A. Lind and B. Marcus, An Introduction to Symbolic Dynamics
and Coding, (1995) Cambridge Univ. Press.}

{\small [34] D. E. Dutkay and P. E. T. Jorgensen, Iterated Function Systems,
Ruelle Operators and Invariant Projective Measures,
arXiv:math.DS/0501077/v3, (2005) Preprint.}

{\small [35] P. E. T. Jorgensen, Use of Operator Algebras in the Analysis of
Measures from Wavelets and Iterated Function Systems, (2005) Preprint.}

{\small [36] D. Guido, T. Isola and M. L. Lapidus, A Trace on Fractal Graphs
and the Ihara Zeta Function, arXiv:math.OA/0608060v1, (2006) Preprint.}

{\small [37] P. Potgieter, Nonstandard Analysis, Fractal Properties and
Brownian Motion, arXiv:math.FA/0701649v1, (2007) Preprint.}

{\small [38] L. Bartholdi, R. Grigorchuk, and V. Nekrashevych, From Fractal
Groups to Fractal Sets, arXiv:math.GR/0202001v4, (2002) Preprint.}

{\small [39] S. Thompson and I. Cho, Powers of Mutinomials in Commutative
Algebras, Undergrad. Research, (2008) St. Ambrose Univ., Dep. of Math.}

{\small [40] S. Thompson, C. M. Mendoza, and A. J. Kwiatkowski, and I. Cho,
Lattice Paths Satisfying the Axis Property, Undergrad. Research, (2008) St.
Ambrose Univ., Dep. of Math.}

{\small [41] T. Shirai, The Spectrum of Infinite Regular Line Graphs, Trans.
AMS., 352, no 1., (2000) 115 - 132.}

{\small [42] J. Kigami, R. S. Strichartz, and K. C. Walker, Constructing a
Laplacian on the Diamond Fractal, Experiment. Math., 10, no. 3, (2001) 437 -
448.}

{\small [43] I. V. Kucherenko, On the Structurization of a Class of
Reversible Cellular Automata, Diskret. Mat., 19, no. 3, (2007) 102 - 121.}

{\small [44] J. L. Schiff, Cellular Automata, Discrete View of the World,
Wiley-Interscience Series in Disc. Math .\& Optimazation, ISBN:
978-0-470-16879-0, (2008) John Wiley \& Sons Press.}

{\small [45] P. E. T. Jorgensen, and M. Song, Entropy Encoding, Hilbert
Spaces, and Kahunen-Loeve Transforms, JMP, 48, no. 10, (2007)}

{\small [46] P. E. T. Jorgensen, L. M. Schmitt, and R. F. Werner, }$q$%
{\small -Canonical Commutation Relations and Stability of the Cuntz Algebra,
Pac. J. of Math., 165, no. 1, (1994) 131 - 151.}

{\small [47] A. Gill, Introduction to the Theory of Finite-State Machines,
MR0209083 (34\TEXTsymbol{\backslash}\#8891), (1962) McGraw-Hill Book Co.}

{\small [49] F. Radulescu, Random Matrices, Amalgamated Free Products and
Subfactors of the $C^{*}$- Algebra of a Free Group, of Noninteger Index,
Invent. Math., 115, (1994) 347 - 389.}

{\small [50] J. E. Hutchison, Fractals and Self-Similarity, Indiana Univ.
Math. J., vol. 30, (1981) 713 - 747.}

{\small [51] E. Schrodinger, What is Life? The Physical Aspect of the Living
Cell, Cambridge Univ. Press, (1944) Cambridge.}
\end{quote}

\end{document}